\documentclass[11pt]{article}
\usepackage[utf8]{inputenc}
\usepackage[a4paper, portrait, margin=1in]{geometry}

\usepackage{color,amsmath,amsfonts,fullpage,amsthm,mathrsfs,fancyhdr,verbatim,framed,hyperref,xspace,mdframed,tabularx,booktabs,enumitem,comment,float, mathtools}
\usepackage[dvipsnames]{xcolor}
\definecolor{linkblue}{HTML}{0072BB}
\usepackage{graphicx}
\hypersetup{
    colorlinks=true,
    linkcolor=blue,
    citecolor=teal
}

\usepackage{fec}
\newtheorem{theorem}{Theorem}

\newtheorem{question}{Question}
\newtheorem{definition}{Definition}
\newtheorem{lemma}{Lemma}

\newtheorem{corollary}{Corollary}

\newtheorem{assumption}{Assumption}
\newtheorem*{inassumption*}{Quantum Input Assumptions}

\usepackage{mathpazo}
\usepackage{graphicx}

\usepackage{braket}
\usepackage{longtable} % in your preamble
\usepackage{tikz}
\usetikzlibrary{shapes.geometric}
\usepackage{pgfplots}
\usepackage{enumitem}

\DeclareMathOperator{\supp}{supp}

\newcommand{\cycnormonew}[2]{\left\|#1\right\|_{\mathrm{cyc},1,#2}}

\def\bra#1{\mathinner{\langle{#1}|}}
\def\ket#1{\mathinner{|{#1}\rangle}}

\newcount\Comments
\newcommand{\kibitz}[2]{\ifnum\Comments=1\textcolor{#1}{#2}\fi}

\newcommand{\expP}{\mathbb{E}_{y \underset{P}{\sim} x}}
\newcommand{\tmix}{t_{\mathrm{mix}}}

\usepackage{authblk}

\title{Quantum Speedups for Group Relaxations of \\Integer Linear Programs}

\author{Brandon Augustino\footnote{These authors contributed equally and are listed in alphabetical order.} \thanks{brandon.augustino@jpmchase.com}}
\author{Dylan Herman\protect\footnotemark[1]} 
\author{Guneykan Ozgul}  
\author{Jacob Watkins}
\author{\\\vspace{-.25cm} Atithi Acharya}
\author{Enrico Fontana}
\author{Junhyung Lyle Kim}
\author{Shouvanik Chakrabarti\thanks{shouvanik.chakrabarti@jpmchase.com}} 
\affil{Global Technology Applied Research, JPMorganChase, New York, NY 10001, USA}

\date{}
\begin{document}

\maketitle

% \begin{abstract}
%     Integer Linear Programs (ILPs) are a flexible and ubiquitous model for discrete optimization problems. Solving ILPs is \textsf{NP-Hard} in general, but of great practical importance. It has proven challenging to identify super-quadratic quantum speedups for ILPs, primarily because most classical algorithms that handle ILPs with many constraints are global and exhaustive, whereas quantum frameworks that offer the potential for super-quadratic speedups leverage the local properties of the objective function and feasible set. We address this difficulty by considering quantum algorithms for Gomory's group relaxation, a relaxation of an ILP that is obtained by removing the nonnegativity constraints from variables that are positive in the optimal solution of the linear programming relaxation, while keeping integrality of the decision variables. We present a classical algorithm that is competitive with known alternatives that solves the group relaxation via a local search, and a corresponding quantum algorithm that under reasonable technical conditions offers a super-quadratic speedup. When the group relaxation satisfies a non-degeneracy condition analogous to (albeit stronger than) that found in linear programming, our approach yields an optimal solution to the original integer program. In other cases, the group relaxation can improve downstream branch-and-cut solvers by reducing the integrality gap, a behavior that we numerically show to be typical for some practically interesting ILPs.  
% \end{abstract}
\begin{abstract}
Integer Linear Programs (ILPs) are a flexible and ubiquitous model for discrete optimization problems. Solving ILPs is \textsf{NP-Hard} yet of great practical importance. Super-quadratic quantum speedups for ILPs have been difficult to obtain because classical algorithms for many-constraint ILPs are global and exhaustive, whereas quantum frameworks that offer super-quadratic speedup exploit local structure of the objective and feasible set. We address this via quantum algorithms for Gomory's group relaxation. The group relaxation of an ILP is obtained by dropping nonnegativity on variables that are positive in the optimal solution of the linear programming (LP) relaxation, while retaining integrality of the decision variables. We present a competitive feasibility-preserving classical local-search algorithm for the group relaxation, and a corresponding quantum algorithm that, under reasonable technical conditions, achieves a super-quadratic speedup. When the group relaxation satisfies a nondegeneracy condition analogous to, but stronger than, LP nondegeneracy, our approach yields the optimal solution to the original ILP. Otherwise, the group relaxation tightens bounds on the optimal objective value of the ILP, and can improve downstream branch-and-cut by reducing the integrality gap; we numerically observe this on several practically relevant ILPs. To achieve these results, we derive efficiently constructible constraint-preserving mixers for the group relaxation with favorable spectral properties, which are of independent interest.
\end{abstract}

\newpage
\tableofcontents
\newpage

\section{Introduction} 

\subsection{Motivation}
Quantum algorithmic speedups for combinatorial optimization have been heavily studied for more than two decades \cite{farhi2000quantum,farhi2014quantum,montanaro2018backtracking,doi:10.1137/1.9781611975482.107,montanaro2020branchAndBound}, both due to the practical importance of these problems, as well as the theoretical optimism emerging from the derivation of many rigorous (usually quadratic) speedups based on amplitude amplification, amplitude estimation, quantum walks, or some combination thereof. However, motivated by the need for larger theoretical quantum speedups if a quantum advantage is to be realized on hardware~\cite{babbush2021focus, omanakuttan2025threshold}, recent years have seen a renewed interest in the development of frameworks for super-quadratic speedups \cite{hastings2018shortPath,dalzell2022mind,chakrabarti2024generalized} and the rigorous characterization of heuristics such as the Quantum Approximate Optimization Algorithm (QAOA)~\cite{farhi2014quantum} and the Quantum Adiabatic Algorithm (QAA)~\cite{farhi2000quantum}.

A unified framework for studying many discrete optimization problems is through the study of \textit{integer linear programs} (ILPs). With $A \in \Zmbb ^{m \times n}$ and $b\in \Z{m}$, define $$\Pcal \coloneqq \{ x \in \Z{n}_{\geq 0} : Ax = b \}.$$ 
Given a vector $c \in \R{n}$, the integer linear programming problem is to solve
\begin{equation}\label{e:ILP}\tag{ILP}
  \textsf{OPT} \coloneqq \min \{ c^{\top} x : x \in \Pcal \}.
\end{equation}
ILP is a classic \textsf{NP-complete} problem underpinning countless applications in computer science, operations research, and decision sciences more generally. It naturally encodes many combinatorial optimization problems, as well as problems from finance, economics, industrial engineering, and the algorithmic geometry of numbers~\cite{schrijver1987polyhedral,grotschel2012geometric,wolsey1999integer,conforti2014integer}. As a consequence of this flexibility, the study of algorithms for solving ILPs has gained significant interest, and there exist several popular commercial~(Gurobi~\cite{gurobi} and CPLEX~\cite{cplex2009v12}) and open-source solvers~(SCIP~\cite{SCIPOptSuite10}). From a practical point of view however, problems at interesting scales can often be solved, even when the asymptotic scaling of the algorithm is exponential. Polynomial speedups over this exponential scaling are significant as they directly expand the domain of problems that can be handled in this way.

Although special instances of ILP are polynomial-time solvable---such as when the feasible region is an integral\footnote{A polyhedron is integral if it equals the convex hull of its integer points.} polyhedron---general instances remain computationally intractable. The difficulty of ILP stems from integrality: it destroys convexity, replacing a well-structured polyhedron with a discrete subset of the integer lattice. Unlike linear programs (LPs), ILPs lack local convexity properties that could inform efficient local search procedures. Even when the number of variables and constraints is small, or the feasible set appears relatively simple (for example, binary ILPs where the decision variables are $x \in \{0,1\}^n$), determining feasibility or optimality can encode well-known \textsf{NP-Complete} problems such as subset-sum and \textsf{3-SAT}. Assuming the exponential time hypothesis~\cite{impagliazzo2001complexity}, these problems cannot be solved in polynomial time. From a practical point of view however, ILPs are solved regularly as part of established computational pipelines. Designing effective algorithms for solving general integer programming problems remains a central and enduring challenge in theoretical and applied computer science.

It is natural to ask what speedups quantum algorithms can provide for ILPs. In the black-box setting an ILP can be made to encode unstructured search, and the largest speedup possible is quadratic due to Grover's search algorithm~\cite{grover1996fast}. For larger speedups (and practical relevance), we must focus on the white-box setting. In this setting, na\"ively applying Grover's algorithm to solve \eqref{e:ILP} would amount to brute-force search in a sufficiently large box $\{x\in\mathbb{Z}^n_{\ge 0}:\|x\|_\infty\le M\}$ for suitable $M$, yielding a runtime $\mathcal{O}((M{+}1)^{n/2}\,\mathrm{poly}(n,L))$, where $L$ denotes the input length of the problem data. This scales with the size of the \emph{ambient} space and ignores structure in the constraints. By contrast, the best (exact) classical ILP algorithms exploit geometry‑of‑numbers to decompose and enumerate the feasible region, and determine the optimal solution using $(\log(2n))^{\Ocal(n)} \cdot \poly (n, L)$ operations~\cite{dadush2012integer,reis2023subspace}. A recent line of work~\cite{montanaro2018backtracking,ambainis2017quantum,montanaro2020branchAndBound,chakrabarti2022universal} has improved upon direct Grover based methods by using quantum walks to accelerate the branch-and-cut algorithms that are used by many integer programming (IP) solvers, but there is no clear path to to obtain the super-quadratic speedups required to realize a quantum advantage in practice~\cite{babbush2021focus,omanakuttan2025threshold}.

A central obstacle to super–quadratic speedups for ILPs is that practically relevant instances typically carry many (often tight) constraints~\cite{abbas2024challenges}. This is unsurprising: in general, the canonical formulation \eqref{e:ILP} does not yield a succinct, algorithmically useful \emph{explicit} description of the feasible region $\Pcal$, making direct constraint handling challenging even classically. Instead, general IP solvers use refined but ultimately exhaustive frameworks which proceed by incrementally tightening \emph{relaxations} of \eqref{e:ILP} to implicitly describe $\Pcal$ rather than maintaining feasibility, and thus cannot exploit the local geometry of the objective on the feasible set. By contrast, the quantum frameworks that promise super‑quadratic gains—short‑path and its generalizations~\cite{hastings2018shortPath,dalzell2022mind,chakrabarti2024generalized}, quantum dynamical algorithms~\cite{chakrabarti2025speedups, herman2025mechanisms}, and heuristics such as QAOA~\cite{farhi2014quantum} and quantum annealing~\cite{farhi2000quantum, kadowaki1998quantum}—are explicitly local and benefit from mixers that preserve (or softly enforce) feasibility. Designing such local exploration for general ILPs is precisely the challenge: classical local methods (e.g., simulated annealing, parallel tempering, tabu search, memetic algorithms) succeed mainly on special cases with few constraints or efficiently enumerable feasible sets; the latter assumption is particularly strong, since even (approximately) uniform sampling over $\Pcal$ is \textsf{NP-Hard}.

This mismatch motivates our central question:
\begin{question}
\label{qstn:main}
    Is there a class of algorithms for integer linear programming that
    \begin{enumerate}
        \item Is \emph{applicable} to general ILPs with many constraints, 
        \item Is based on \emph{local} exploration of a feasible space, and as a consequence,
        \item Is a viable candidate for \emph{super-quadratic} quantum speedups?
    \end{enumerate}
\end{question}

In this paper, we answer Question~\ref{qstn:main} in the affirmative and give the first, to our knowledge, analysis of a framework that satisfies all the above conditions. Central to our results is a relaxation of \eqref{e:ILP} introduced by Gomory~\cite{gomory1965relation,gomory1967faces, gomory1969some} termed the \textit{group relaxation} (also known as the \textit{group problem}). The next sub-section introduces relaxations of ILPs in general and Gomory's group relaxation more specifically, and discusses its significance and applicability to ILP solving. We then go on to outline our main contributions, and discuss their role in satisfying the criteria of Question~\ref{qstn:main}.

\subsection{Relaxations of ILPs}
Modern integer-programming solvers—commercial (Gurobi~\cite{gurobi}, CPLEX~\cite{cplex2009v12}) and open-source (SCIP~\cite{SCIPOptSuite10})—follow a relax–tighten–search paradigm. A relaxation of \eqref{e:ILP} is a problem obtained by dropping some of the constraints in the original formulation. Relaxations are thus problems of the form $\OPT^{\prime} \coloneqq \min \{ c^{\top} x : x \in \Pcal^{\prime} \}$, with $\Pcal \subseteq \Pcal^{\prime}$ and $c^{\top} x \geq \OPT^{\prime}$ for any $x \in \Pcal$. The challenge is to design relaxations that are both tractable and useful, since there is a natural tradeoff between the computational effort required to solve the relaxation, and the quality of the lower bound it provides for $\OPT$. 

The canonical choice is the LP relaxation, obtained by dropping integrality:
\begin{equation}\label{e:LP}\tag{LP}
  \textsf{OPT}_{\text{LP}} \coloneqq \min \left\{ c^{\top} x : x\in \Pcal_{\text{LP}} \right\},
\end{equation}
where $$\Pcal_{\text{LP}} \coloneqq \{ x \in \R{n}_{\geq 0} : Ax = b \}.$$ 
The upshot of working with the LP relaxation is that it can be solved efficiently using \textit{interior point methods}~\cite{wright1997primal}, and clearly $\OPT_{\text{LP}} \leq \OPT$, since $\Pcal := \Pcal_{\text{LP}} \cap \Z{n}$. If the optimal solution to \eqref{e:LP} is already integral, it automatically certifies optimality for the ILP. Unfortunately, aside from very restrictive classes of ILPs, such as those in which the constraint matrix $A$ is \textit{totally unimodular}\footnote{A matrix $A$ is totally unimodular (TU) if every square submatrix has determinant in $\{-1,0,1\}$. When $A$ is TU and $b$ is integral, the polyhedron $\{x\in\mathbb{R}^n_{\ge 0}:Ax=b\}$ is integral, so the LP relaxation attains an integer optimum. TU matrices arise in network and flow models, including bipartite matching/assignment, transportation, and min‑cost flow.}, this is rarely the case.  Moreover, a simple rounding of solutions to LP relaxations typically fails to yield high quality solutions to \eqref{e:ILP} without incurring exponential overhead.  

Alternatively, the first nontrival integer programming algorithms proposed combining information obtained from the solution to the LP relaxation with logical disjunction to describe and enumerate $\Pcal$. Given an optimal solution $x_{\text{LP}}^{\star}$ to \eqref{e:LP}, \textit{cutting plane algorithms}, which were first introduced by Gomory~\cite{gomory1958outline}, generate an improved formulation of \eqref{e:ILP} whenever $x_{\text{LP}}^{\star} \notin \Pcal$. This tightened formulation is obtained by adding a linear constraint that is satisfied by every point in $\Pcal$, but is not obeyed by $x_{\text{LP}}^{\star}$. Such an inequality is called a \textit{cut}. This process continues until enough cuts have been added so that the resulting LP relaxation has an integral optimal solution, thereby solving \eqref{e:ILP}. In Land and Doig's algorithm~\cite{land1960automatic}, whenever $x_{\text{LP}}^{\star} \notin \Pcal$ the feasible region of the integer program is divided into subregions  satisfying $x_{\text{LP}}^{\star} \notin \Pcal_1$ and $x_{\text{LP}}^{\star} \notin \Pcal_2$. A canonical example is to branch on a fractional coordinate $x_{\text{LP}, j}^{\star} \notin \Z{}$ by considering $\Pcal_1 := \{x \in \Pcal : x_j \leq \lfloor x_{\text{LP}, j}^{\star} \rfloor \}$ and $\Pcal_2 := \{x \in \Pcal : x_j \geq \lceil x_{\text{LP}, j}^{\star} \rceil \}$. This process continues until a global certificate of optimality (or infeasibility) is obtained. Over the years, many different strategies for decomposing $\Pcal$ have been developed along these lines, and the resulting family of methods are termed as \textit{branch-and-bound}; see \cite{achterberg2007constraint} for a modern survey. Using cutting planes within branch-and-bound gives rise to \textit{branch-and-cut}. Variants of these algorithms underlie the software in \cite{gurobi, cplex2009v12, SCIPOptSuite10}. 

Stronger relaxations are therefore central. Working with a relaxation whose feasible region more closely approximates the original feasible region $\Pcal$ produces tighter bounds on $\OPT$, and reduces the search footprint. Gomory’s \textit{group relaxation}~\cite{gomory1965relation,gomory1967faces, gomory1969some} is a principled strengthening derived from an optimal LP basis. Assume $\rank(A) = m$ and let $\Bcal  \subset [n]$ denote an optimal basis of the LP relaxation. We can partition the variables and data according to $\Bcal $ and its complement $\mathcal{N} \coloneqq [n] \setminus \Bcal $, writing
$$ A = [A_\Bcal, A_\Ncal], \quad c = [c_\Bcal, c_\Ncal], \quad x = [x_\Bcal, x_\Ncal].$$
The group relaxation is formed by dropping the nonnegativity constraints on $x_\Bcal$ (i.e., the entries that took positive values in the optimal solution of the LP relaxation) in \eqref{e:ILP},  leading to:
\begin{equation}\tag{$\text{ILP}_{\Bcal }$}\label{e:group_prob_intro}
       \OPT_\Bcal \coloneqq  \min \left\{c^{\top} x: x \in \Pcal_{\Bcal} \right\},
\end{equation}
where 
$$\Pcal_{\Bcal} := \{ (x_{\Bcal} , x_{\Ncal}) \in \Z{m} \times \Z{n-m}_{\geq 0} : Ax = b \}.$$

It is immediate that
$
\OPT_\Bcal \le \OPT,
$
since $ \Pcal = \Pcal_{\Bcal} \cap \Zmbb _{\ge 0}^n$. Note that $ \Pcal_{\Bcal} \not\subset \Pcal_{\text{LP}}$ in general, because the group relaxation permits $x_{\Bcal}$ to take negative values. However,
$
\OPT_{\text{LP}} \le \OPT_\Bcal
$. To see this, note that for $(x_{\Bcal},x_{\Ncal}) \in \Pcal_{\Bcal}$, we may use invertibility of $A_{\Bcal}$ and write
$$
c_\Bcal^\top x_\Bcal +c_\Ncal^\top x_\Ncal
= c_{\Bcal}^\top A_{\Bcal}^{-1}b + \big(c_{\Ncal} - A_{\Ncal}^\top (A_{\Bcal}^{-1})^\top c_{\Bcal}\big)^\top x_{\Ncal},
$$
where $\bar{c}_{\Ncal} \coloneq c_{\Ncal} - A_{\Ncal}^\top (A_{\Bcal}^{-1})^\top c_{\Bcal} \geq 0$ are the usual LP \textit{reduced costs} for the nonbasic indices. It is well-known~\cite{schrijver1998theory, wolsey1999integer} that, upon eliminating $x_{\Bcal}$ via $x_{\Bcal} = A_{\Bcal}^{-1}(b - A_{\Ncal}x_{\Ncal})$, the group relaxation \eqref{e:group_prob_intro} can be equivalently expressed as
\begin{equation}\label{group_prob_intro_algebraic}
\OPT_\Bcal
= \min_{x_{\Ncal} \in \mathbb{Z}_{\geq 0}^{n-m}} \left\{\, c_{\Bcal}^{\top} A_{\Bcal}^{-1} b + \bar{c}_{\Ncal}^{\top} x_{\Ncal} : A_{\Bcal}^{-1}A_{\Ncal}x_{\Ncal} \equiv A_{\Bcal}^{-1}b \pmod{\Zmbb^m} \,\right\}.
\end{equation}  
% where $Z \coloneqq \bigoplus_{j=1}^{n-m} \Zmbb_{s_j}$, and $s_j$ is the order of the $j$-th column of $A_\Ncal$ in the lattice generated by $A_{\Bcal}$. 
Since $\Bcal$ is optimal $c_{\Bcal}^\top A_{\Bcal}^{-1}b  = \OPT_{\text{LP}}$, which combined with nonnegativity of $\bar{c}_{\Ncal}$ asserts that feasible solutions to \eqref{group_prob_intro_algebraic} satisfy
$$
c_{\Bcal}^\top A_{\Bcal}^{-1}b + \bar{c}_{\Ncal}^\top x_{\Ncal} \geq \OPT_{\text{LP}}.
$$
As a consequence, $ \OPT_\Bcal \ge \OPT_{\text{LP}}$ and together, we have the chain
$
\OPT_{\text{LP}} \le \OPT_\Bcal \le \OPT.
$

The group relaxation can, under suitable conditions, solve the original ILP to optimality. Gomory originally showed~\cite{gomory1965relation} that \eqref{e:group_prob_intro} is an ``asymptotic relaxation'': it recovers the ILP optimum whenever the right‑hand side $b$ is sufficiently large~\cite[Theorem 5.6]{wolsey1999integer}. More generally, Lasserre~\cite{lasserre2004generating, lasserre2009linear} identified a discrete nondegeneracy condition under which the group relaxation is exact: $ \OPT_\Bcal = \OPT$ unless there exists another optimal LP basis $ \Bcal' \neq \Bcal$ with $ \OPT_{\Bcal'} = \OPT_\Bcal$. A certificate of optimality can also be obtained efficiently. If a group‑optimal solution $x^\star \in \Pcal_{\Bcal}$ happens to be feasible for the original ILP, then $x^\star$ is optimal for the ILP because $ \OPT_\Bcal \le \OPT$.

Even when the group relaxation is degenerate (so $\OPT_\Bcal < \OPT$), it still provides strictly stronger information than the plain LP relaxation. As a bounding oracle, it yields tighter lower bounds, often allowing branch‑and‑bound to fathom nodes that would pass an LP bound test. As a tightening oracle, the \textit{corner polyhedron}~\cite{gomory1965relation,gomory1967faces, gomory1969some}—the convex hull of group‑feasible points—has been tremendously impactful in integer programming due to its
connection to cutting planes~\cite[Chapter 6]{conforti2014integer}. The geometry of the corner polyhedron exposes facet inequalities valid for the original ILP that typically dominate standard LP‑derived cuts, further shrinking the integrality gap~\cite[Chapter 19]{junger200950}.

Finally, while $\Pcal_{\Bcal}$ (and hence the feasible set of \eqref{group_prob_intro_algebraic}) can appear to be infinite as written, it is known that whenever the group relaxation is feasible there exists an optimal solution with nonbasic coordinates bounded by their orders, i.e., $x^{\star}_j \in \{0,\dots,u_j-1\}$ for $j\in\Ncal$, where $u_j$ is determined by the lattice generated by $A_{\Bcal}$~\cite{shapiro1968dynamic, lasserre2004generating}. Accordingly, one may work over an \emph{optimum‑preserving} finite feasible subset
$$
\overline{\Pcal}_{\Bcal} \coloneqq \left\{\, x_{\Ncal} \in \bigoplus_{j=1}^{n-m} \mathbb{Z}_{u_j} \;:\; A_{\Bcal}^{-1}A_{\Ncal}x_{\Ncal} \equiv A_{\Bcal}^{-1}b \pmod{\Z{m}} \,\right\}.
$$
% \dyh{So do we want to change this? Technically, depending on whether we use the $r$'s or $s$'s the below thm referencing this could be incorrect. although, it's informal so maybe it's ok. Also, the later sections are currently written as if the finite version as not been introduced yet. We can chat about what to do too.}\ba{So I was just trying to provide high-level intuition that informs the reader why we can work over a bounded set. I am happy to change this further, but I changed the s's to $u$ so that it could be kept informal}
In this bounded‑residue view, feasibility is captured by a linear congruence over a finite abelian group and the feasible points form a coset. For more details, see Section~\ref{sec:primal}. This structure is well suited to quantum primitives—finite abelian groups admit efficient quantum Fourier transforms~\cite{kitaev1995quantum}, and uniform superpositions over residue classes can be prepared coherently.

\subsection{Contributions}
The main contributions in this paper can be summarized as follows:  
\begin{enumerate}
    \item We give a competitive classical algorithm for solving Gomory’s group relaxation via a feasibility–preserving local search over a finite abelian group.
    \item We design a quantum algorithm within the generalized short–path framework~\cite{chakrabarti2024generalized}, provide efficient state–preparation and block–encodings, and identify conditions under which the speedup is \emph{super–quadratic}.
\item We construct \emph{constraint–preserving mixers} over the group relaxation’s feasible set—(i) a product–of–cycles with explicit log–Sobolev bounds and (ii) an expander with constant spectral gap. We also show that the ground states of the corresponding mixing Hamiltonians can be prepared efficiently. Our approach therefore provides feasibility‑preserving, good‑gap primitives for short‑path/QAOA/QAA on group relaxations; and these are of independent interest. 
\end{enumerate}

We begin first by informally stating our main result as follows.
\begin{theorem}[Main Result, informal statement of Theorems \ref{thm:main-speedup} \& \ref{thm:main-speedup-poincare}]
\label{thm:main-informal}
Let $\Bcal$ be an optimal basis for the LP relaxation of \eqref{e:ILP}. Define the set of optimal solutions to the (finite) group relaxation associated with this basis to be 
$$\overline{\Pcal}_{\Bcal}^{\star} := \{ x^{\star}_{\Ncal} \in \overline{\Pcal}_{\Bcal} : \OPT_{\textup{LP}} + \bar{c}_{\Ncal}^{\top} x^{\star}_{\Ncal}  = \OPT_\Bcal\}.$$
We construct a quantum algorithm that, with high probability, finds an optimal solution $x^{\star}_{\Ncal} \in \Pcal_\Bcal^\star$ with runtime 
$$
    \mathcal{O} \left( \left( |\overline{\Pcal}_{\Bcal}|/|\overline{\Pcal}_{\Bcal}^\star| \right)^{\frac{1}{2}-\alpha(n)} \;\poly (n, L) \right),$$
where
$\alpha (n) > 0$ and satisfies $\alpha (n) = \Ocal_n(1)$. Under mild technical assumptions, $\alpha (n) = \Theta_n(1)$, in which case the speedup is genuinely super-quadratic.  
\end{theorem}
This theorem combines each of our main contributions. Firstly, as discussed, the classical algorithm that we introduce locally explores the feasible space of the group relaxation and hence incurs a cost of $\Ocal \left((|\overline{\Pcal}_{\Bcal}|/|\overline{\Pcal}_{\Bcal}^\star| ) \; \poly (n, L) \right)$. Our quantum algorithm is therefore an improvement over a direct Grover-like speedup of this classical algorithm. In fact, this improvement occurs under the technical conditions that make $\alpha(n) = \Theta_n(1)$. These conditions will be discussed in more detail in the technical sections, and we will give constructive examples where they are satisfied. 

We also note that each of these developments are in fact \emph{necessary} to obtain this improvement. The framework of \cite{chakrabarti2024generalized} as well as several other quantum frameworks that may enable super-quadratic speedups require a ``mixer" Hamiltonian that drives the exploration of the feasible space. In the following subsections, we explore the main technical ideas behind each contribution.

\paragraph{Solving the group relaxation with null space walks} 
As discussed in our motivations, quantum algorithms for discrete optimization can struggle to handle constraints in a natural way. For example, a large class of methods are based on quantum Hamiltonians, which encode the cost function in a diagonal Hermitian operator and perform ``mixing" operations, via a non-diagonal Hamiltonian, to transform from some initial configuration to the optimal solution, i.e. a minimal energy configuration of the diagonal Hamiltonian. To handle constraints, a penalty term is often added to the objective, or, in situations where the constraint set is simple (e.g., a single Hamming weight constraint) a constraint-preserving mixer is used. This typically has detrimental effects on performance, however, since either approach often disrupts the spectral properties of the relevant Hamiltonians which determine the efficiency of such approaches. This is connected to a similar difficulty for \emph{classical} algorithms, where most algorithms that handle constraints are global and exhaustive. Local algorithms that explore the feasible space directly, either resort to penalty terms or are limited to constrained spaces that can be efficiently ``enumerated." Since the enumeration of the feasible space of an ILP is \textsf{NP-complete} in general, this is only applicable to very specific ILPs.

The first ingredient of our quantum algorithm for the group relaxation is therefore a corresponding classical algorithm that is only based on a local search.
Classical methods for solving the group relaxation are rooted in Gomory’s original work, and run in time $\Ocal(n \; |\det A_{\Bcal }|)$, see, e.g.,~\cite[Proposition~4.4]{aardal2002non}. These involve searching for the shortest path between two elements (the identity and an element determined by the right-hand side vector $b$) in a directed Cayley graph associated with a particular finite abelian group $G$, where the presented runtime uses the upper bound $\lvert G \rvert \leq |\det A_{\Bcal }|$. The group $G$ is related to the integer span of the constraint matrix $A_{\Bcal}^{-1}A_{\Ncal}$. This algorithm does not directly explore the feasible space locally; it actually always first encounters infeasible points before finding the solution. Nevertheless, the algorithm does use more local information than an exhaustive search method. There is also no current quantum algorithmic framework that facilitates a path to super-quadratic speedups for this shortest-path approach, and we leave this as an interesting potential direction for future work. 

To ameliorate these difficulties, we move the search to the null space of the constraint matrix, $A_{\Bcal}^{-1}A_{\Ncal}$. Specifically, we reduce \eqref{e:group_prob_intro} to a problem with a domain $Z$ that is a finite abelian group. The feasible set for this finite problem is a coset $\overline{\Pcal}_{\Bcal} = \hat{\mathbf{x}}+K \subseteq Z$, where $\hat{\mathbf{x}}\in Z$ is any feasible residue vector and $K \leq Z$ is the kernel of the map $ \zmbf \mapsto A_{\Bcal }^{-1}A_{\mathcal{N}}\zmbf \pmod{\Z{m}}$, $\zmbf \in Z$. Note that $|\hat{\mathbf{x}}+K|=|K|$. In Section~\ref{sec:primal}, we construct a random walk on $K$, where each step of the walk remains within the feasible set by design. This approach ensures that feasibility is preserved at every step, and allows us to efficiently prepare a uniform superposition over feasible solutions to initialize the algorithm. We show in Lemma~\ref{lem:kernel_gen_finding} that an arbitrary generating set for $K$ can be found in $\mathcal{O}\left(\textup{poly}(n, L)\right)$ classical time. 

% Conversely, our quantum algorithm performs a random walk over the kernel subgroup $K \le Z$ that preserves feasibility, where $Z$ is the finite ambient space group, and obeys $|Z| = |G|\; |K|$. Our quantum algorithm thus yields a super-quadratic speedup whenever $|K|/|K^\star|$ is not too large relative to $|G|$, where $K^\star$ is the set of optimal solutions in the feasible coset.
% \sh{Could someone with a bit more context add a paragraph here about Dylan's construction and its generality. We should also comment on why it is reasonable to expect this to be competitive with Gomory's approach in a reasonable setting.} 
It remains to argue that the local search algorithm is comparable to preexisting classical methods. Particularly, our algorithm is competitive with Gomory's shortest-path approach when $\lvert K \rvert$ is comparable to $\lvert G \rvert$.
%Since we are searching the null space via random walk, we would require $|K|$ to not be too large relative to $|G|$. 
In Section~\ref{s:compare}, we demonstrate a family of ILPs that satisfy this condition. However, one can determine this for any particular problem efficiently.

Equipped with our local search algorithm for the group relaxation, we are now ready to explore corresponding quantum speedups. While we rely on the framework of \cite{chakrabarti2024generalized} to obtain rigorous runtimes, we point out that this walk is potentially useful to derive mixing Hamiltonians for a variety of other quantum algorithms. These include popular heuristics such as the QAOA and quantum annealing, as well as algorithms such as the quantum adiabatic algorithm whose runtime analysis diverges from that considered here.

% \ba{May want to mention here: The nullspace walk may be of independent interest as a means for handling constraints/tackling ILPs with QAOA, QAA...}\guney{Algorithms like QAOA and QAA require some sort of local mixer to work which is what nullspace walk is providing. I think this would be interesting to point out as an additional contribution.  }\dyh{Yeah this also goes back to a comment i added above. Here we are technical able to tackle problems that you can't enumerate, ILPs (i.e. modulo the group problem being an approximation to the real ILP). This allows us to start with ``approximately feasible'' yet more non-uniform states. Contrast with QAOA for MIS starting from empty set. I know this is a stretch but might be worth mentioning.}

\paragraph{Generalized Short Path algorithm}

Our work extends the generalized short path\footnote{Note that here and throughout, ``short path'' has nothing to do with the shortest path problem.} (SP) algorithm of~\cite{chakrabarti2024generalized}, which itself builds on the SP framework of Hastings~\cite{hastings2018shortPath} and Dalzell et al.~\cite{dalzell2022mind}. These algorithms refine the quantum adiabatic algorithm (QAA)~\cite{farhi2000quantum} by introducing a piecewise-linear transformation of the cost Hamiltonian and replacing slow adiabatic evolution with two quantum ``jumps'' between ground states. Let $\Xcal$ be a finite set, and consider the optimization problem $\min_{x \in \Xcal} f(x)$, encoded as a quantum Hamiltonian $H$ with $H \ket{x} = f(x)\ket{x}$. Let $E^\star<0$ denote the ground state energy, i.e., the optimal objective value $f(x^\star)$. The generalized SP Hamiltonian is
\begin{equation}
    H_{\mu} = -D(P) + \mu \,\theta_{\eta} (H/|E^\star|),
    \label{eq:SP-hamiltonian}
\end{equation} 
where $\theta_\eta(x) \coloneqq \min\{0, (x + 1 - \eta)/\eta\}$ is a piecewise-linear function, $\mu > 0$ is a tunable parameter, and $D(P)$ is the discriminant matrix of the Markov chain over the feasible set (see Definition~\ref{defn:discriminant}). The algorithm then starts by preparing the ground state of the discriminant matrix, and then prepares the ground state $\ket{\psi_\mu}$ of an intermediate Hamiltonian $H_{\mu}$ where the chosen $\mu$ depends on the properties of the random walk and the cost function. The key idea is that the probability measure $|\psi_{\mu}|^2$ is more concentrated on the optimal solutions than the stationary distribution of the classical random walk. By using quantum amplitude amplification, this algorithm overall uses super-quadratically fewer queries to the cost function $H$ than existing classical algorithms. 

Using the generators of the null space $K$, one can instantiate a feasibility-preserving product–of–cycles Cayley walk with explicit spectral gap bounds. One can also leverage the nice expansion properties of Cayley graphs, and apply the \textit{Alon-Roichmain theorem}~\cite{alon1994random} to construct an expander walk with spectral gap lower bounded by a constant. Both walks yield runtime guarantees that scale with the number of feasible solutions, and achieve a constant exponent improvement over direct amplitude amplification on $K$ under mild technical conditions. For each walk, we show how to prepare the corresponding ground states efficiently, and construct block‑encodings for both $D(P)$ and $H$. The result is an end‑to‑end, feasibility‑preserving pipeline for group relaxations that advances short path (and, by extension, QAOA/QAA) beyond prior unconstrained‑centric applications.

\paragraph{Discussion} It is clear that our result satisfies most of the criterion laid out in Question~\ref{qstn:main}. Specifically, the group relaxation can be applied to general ILPs with many constraints, the new classical algorithm is based on local search, and a super-quadratic quantum speedup can be obtained. 

The only remaining question is that of \emph{usefulness}. After all, Theorem~\ref{thm:main-informal} only offers us a way to solve a relaxation of the ILP, and not the ILP itself, and it is not immediately clear whether there is an impact for the latter problem. 
Solving the group relaxation can be helpful for solving the ILP in two ways. Firstly, when the group relaxation is nondegenerate (see, Definition~\ref{def:non_degen}), the group relaxation solves the ILP directly. In this case, the new classical algorithm that we are accelerating is in fact solving the ILP itself. Whether the quantum algorithm achieves a quantum speedup over \emph{state-of-the-art} classical algorithms for such ILPs depends not only on the size of the feasible set $|\overline{\Pcal}_{\Bcal}|$, but on the performance of state-of-art solvers for these problems (it could be the case, for example that ILPs with nondegenerate group relaxations are easier for these solvers). 

A different setting is when the group relaxation is used as the bounding procedure for an outer Branch-and-Bound algorithm. It is known that the performance of Branch-and-Bound algorithm is heavily influenced by the strength of the bounding procedure that is used at each node, and a stronger bounding procedure can substantially reduce the number of nodes that must be explored by the branching procedure. A quantum speedup for the group relaxation can therefore indirectly accelerate a branch-and-bound solver if the resulting optimality gap is lesser than that obtained from the continuous LP relaxation. In Section~\ref{sec:numerical}, we test this hypothesis on a mixture of benchmark and synthetically generated integer problems. We find that the group relaxation indeed substantially reduces the gap compared to the LP relaxation. Furthermore, we observe that for a substantial fraction of tested problems, solving the group relaxation produces a solution for the original ILP.

Beyond pure ILPs, the group relaxation machinery extends naturally to \emph{mixed‑integer linear programs} (MILPs), which are problems of the form:
\begin{equation}\label{e:MILP}\tag{MILP}
\textsf{OPT}_{\text{MILP}} \coloneqq \min \left\{ c^{\top} x : Ax = b,\; x \in \mathbb{Z}_{\ge 0}^{p} \times \mathbb{R}_{\ge 0}^{\,n-p} \right\}.
\end{equation}
MILPs blend integer and continuous decisions and thus enable one to model a wide range of challenging applications. Wolsey~\cite{wolsey1971extensions} showed that, after clearing denominators in the continuous block by choosing
$$
\mathbf{D} \coloneqq \mathrm{lcm}\big\{|\det A_{\Bcal}| : \Bcal \in \mathscr{B} \cap \{p{+}1,\dots,n\}\big\}
$$
and substituting $z_j=\mathbf{D}\,x_j$ for the continuous variables, every optimal solution to \eqref{e:MILP} is obtained by solving the resulting ILP. This reduction allows the group relaxation techniques developed here to apply verbatim. That said, computing a practically useful $\mathbf{D}$ (and the requisite bases) at scale is nontrivial; devising scalable strategies for this step is beyond the scope of this work and left to future investigation.

\subsection{Related work}

\paragraph{Integer linear programming}
We give a high-level overview of algorithmic development for ILP. For a thorough treatment of the field, we refer the reader to the classic texts \cite{grotschel2012geometric,wolsey1999integer,schrijver1987polyhedral,conforti2014integer}. Gomory's cutting plane method was the first algorithm for solving (mixed) integer programming problems~\cite{gomory1958outline}. Papadimitriou~\cite{papadimitriou1981complexity} proved that ILP was in \textsf{NP}, and as a by product, gave a pseudopolynomial-time dynamic programming algorithm for ILPs with a fixed number of constraints. 

A significant breakthrough was made by Lenstra~\cite{lenstra1983integer}, who showed that ILPs are polynomial time solvable when the number of variables is fixed. A careful complexity analysis shows that Lenstra's algorithm runs in time $2^{\Ocal(n^2)}$ 
times a polynomial in the encoding size $L$ of the problem data. Kannan improved the leading factor to $n^{\Ocal(n)}$ shortly after \cite{kannan1983improved}. Both algorithms fundamentally rely on lattice point enumeration techniques. \textit{Lenstra-type} algorithms utilize a thinnest direction to decompose the IP into lower dimensional subproblems which correspond to parallel lattice hyperplanes. \textit{Kannan-type algorithms} generalize Lenstra-type algorithms by decomposing the feasible region along lattice shifts of a linear subspace. 

Building on these ideas, Dadush~\cite{dadush2012integer} developed an improved Kannan-type algorithm for the more general problem of \textit{convex integer programming}, where the task is to decide whether a convex set $\Kcal \subseteq \R{n}$ contains an integer point. In his thesis, Dadush showed that his algorithm runs in time $2^{\Ocal(n)} n^n \cdot \poly (L)$. He further observed that, assuming a conjecture of Kannan and Lov\'asz~\cite{kannan1988covering} on the best volume‑based lower bound for the covering radius of a convex body with respect to a lattice\footnote{Intuitively, the covering radius is the smallest scaling of a convex body that guarantees a lattice point in every translate.}, the overall complexity could be tightened to $(\log(2n))^{\Ocal(n)} \cdot \poly(L)$. In a recent breakthrough, Reis and Rothvoss~\cite{reis2023subspace} gave a constructive resolution of the \textit{subspace flatness conjecture}, thereby enabling Dadush’s algorithm to attain this improved bound.

The running time of the algorithm described in \cite{reis2023subspace} cannot be reduced in general, and whether there exists a single exponential time algorithm for exact ILP remains open.\footnote{Kanann and Lova\'sz~\cite{kannan1988covering} give an instance that saturates the $\log (n)$ lower bound for the covering radius.} Alternatively, Dadush, Eisenbrand and Rothvoss~\cite{dadush2025approximate} give an algorithm based on reducing the exact ILP problem to a sequence of approximate ones. These intermediate problems are solved using an algorithm for approximate IP from Dadush~\cite{dadush2014randomized}, which runs in time $2^{\Ocal(n)}$. This approach leads to a framework for exactly solving general ILPs in time $2^{\Ocal(n)} \cdot n^n$, but matches the state of the art $(\log (2n))^{\Ocal(n)}$ complexity bound for problems with solutions whose entries are at most $\poly (n)$.

\paragraph{Quantum algorithms for integer and combinatorial optimization}
By now quantum proposals for combinatorial optimization have been investigated for over two decades. The Quantum adiabatic algorithm~\cite{farhi2000quantum} and Quantum Approximate Optimization algorithm~\cite{farhi2014quantum} are the most thoroughly studied frameworks. 

There are also universal quantum speedups for backtracking~\cite{montanaro2018backtracking,ambainis2017quantum}, which themselves furnish universal quantum speedups for branch-and-bound~\cite{montanaro2020branchAndBound,chakrabarti2022universal}. However, the speedup here is quadratic. More recently, Decoded Quantum Interferometry~\cite{jordan2024optimization} has been proposed for solving optimization problems over finite fields. This approach is based on Regev's reduction~\cite{regev2009lattices}, and obtains exponential speedups for approximately solving the \textit{optimal polynomial intersection} problem. 
%\sh{Point out that to the best of our knowledge this is the first super-quadratic speedup for a ``general" integer programming problem?}

\paragraph{The group relaxation and the corner polyhedron}
Gomory introduced and studied the group relaxation and the corner polyhedron in \cite{gomory1965relation,gomory1967faces,gomory1969some}. Wolsey~\cite{wolsey1971extensions} generalized the group relaxation to \textit{mixed integer programming} shortly after, and introduced \textit{extended relations}. Extended relations are all of the possible group relaxations one can get upon dropping successively fewer nonnegativity constraints. One of the extended relaxations is guaranteed to solve the ILP to optimality, since the original ILP trivially corresponds to the extended relaxation obtained from dropping no non-negativity constraints. 

There are several studies that investigate strategies for mitigating poor performance in the degenerate setting~\cite{bell1977convergent,gorry1973computational,wolsey1973generalized,wolsey1999integer}. Empirical studies on the tightness of the group relaxation and the closely related \textit{corner relaxation} are provided in~\cite{fischetti2008tight,cornuejols2012tight}.

In addition to characterizing the nondegeneracy condition, Lasserre~\cite{lasserre2004generating} showed that group relaxations are related to a generalized residue formula from Brion and Vergne~\cite{brion1997lattice}, which serves as the foundation for efficient lattice point enumeration in polytopes when the dimension is a fixed constant~\cite{barvinok1994polynomial,barvinok1999algorithmic,de2004effective}. 

\subsection*{Paper organization}
The rest of this paper is organized as follows. Section~\ref{sec:prelim} establishes notation and reviews the necessary background material, covering finite abelian groups, lattices, Markov chains, linear programming and the generalized short–path framework from \cite{chakrabarti2024generalized}. In Section~\ref{sec:primal}, we develop the null space formulation of the group relaxation, and present a classical Markov–chain local search over the kernel. In Section~\ref{s:quantum}, we give the quantum algorithm: we specify state preparation and block–encodings, and analyze two mixers (product–of–cycles and expander walks) to obtain our main runtime guarantees. Section~\ref{sec:numerical} reports empirical results on applying the group relaxation to synthetic cutting–stock instances and every pure ILP found in MIPLIB~2017~\cite{miplib2017}. Appendix~\ref{app:kernel_compression} details practical strategies for optimizing the null space computations, and Appendix~\ref{sec:alt_speedup} discusses an alternative speedup avenue (Gibbs sampling) and how this setting interacts with our assumptions.  

\section{Preliminaries}\label{sec:prelim}
Let $[d] = \{ 1, \dots, d\}$, and $\R{n}_{\geq 0}$ (resp. $\R{n}_{> 0}$) denote the nonnegative (resp. positive) orthant. We represent the $n \times 1$ all-ones (resp. all-zeros) vector by $\mathbf{1}_n \in \R{n}$ (resp. $\mathbf{0}_n \in \R{n}$). Let $I_m$ denote the $m\times m$ identity matrix. Given a matrix $A$, let $A_j$ refer to the $j$th column of $A$. More generally, given a subset $S \subseteq [n]$ of indices, let $A_S$ denote the submatrix whose columns are indexed in $S$. Similarly, a vector $v$ has subvectors $v_S$ obtained by collecting only those entries indexed in $S$. Given a vector over $\Rmbb$, or equivalently a real-valued function whose domain is finite, the support $\supp(v)$ of $v$ is the collection of indices such that $v_i \neq 0$. We let $\|\cdot\|_p$ denote the usual $\ell_p$-norm on $\mathbb{R}^d$ for $p \in [1,\infty]$, and the $\ell_0$-norm counts the number of non-zero entries. We use $\mathbf{e}_j \in \mathbb{R}^{n}$ to denote the $j$-th, $n$-dimensional standard basis vector.

\subsection{Abelian groups and lattices}\label{subsec:groups_lattices}

For the basic definitions and properties of groups, we refer the reader to any abstract algebra textbook such as~\cite{lang2012algebra}. Here we provide some key results and definitions pertinent to this work.
 
A group $G$ is called \emph{abelian} if its group binary operation is commutative. In the abelian setting, it is customary to represent the group operation using the plus symbol, its identity using $0$, its inverse using a minus sign, etc. All groups considered in this work are abelian. For an integer $k$ and $g \in G$, let $k g$ denote repeated addition of $g$ or its inverse $k$ times, depending on the sign of $k$. Let $\cong$ denote group isomorphism. We will use the notation $S \leq G$ to denote subgroups $F=S$ of $G$, and $G/S = \{g + S : g \in G\}$ to denote the quotient group, whose elements are cosets $g + S = \{g+s : s\in S\}$. Given $a, b \in G$, we write $a = b \pmod{S}$ if $a \in b + S$. To maintain brevity, we frequently use $a$ to refer to the coset $a + S$, and context should make clear when the quotient space is being discussed.

Given a subset $S \subseteq G$, let $\langle S\rangle$ denote the smallest subgroup of $G$ containing $S$. We say that $S$ \emph{generates} $G$ if $\langle S \rangle = G$. Equivalently, $S$ generates $G$ if every $g \in G$ can be obtained from finite products and inverses of elements in $S$. A group $G$ is called \emph{cyclic} if it has a single generator $g$, i.e., $G = \langle g\rangle$. The \emph{order} of a group $G$ is simply its cardinality as a set $\abs{G}$. We also refer to the order of an element $g$ to mean the order of $\langle g\rangle$.

Let $\Zmbb_r$ denote the additive group of integers modulo $r$. The \emph{fundamental theorem of finite abelian groups} states that every nontrivial finite abelian group $G$ is isomorphic to a direct sum 
\begin{align*} 
    G \cong \Zmbb _{r_1} \oplus \Zmbb _{r_2}\cdots \oplus \Zmbb _{r_k},
\end{align*}
for integers $r_1,\ldots, r_k \geq 2$. There are several canonical choices for the integers $r_j$; here we typically consider the \emph{invariant factor} decomposition, whereby $r_j \vert r_{j+1}$ (i.e., $r_j$ divides $r_{j+1}$), for all $j\in [k-1]$. In this case, the $r_j$ are called invariant factors, and the decomposition is unique.

A \emph{norm} $\norm{\cdot}:G \rightarrow \Rmbb_{\geq 0}$ over an abelian group $G$ is a real-valued function that is positive definite, subadditive, and inverse symmetric ($\norm{-g} = \norm{g}$). Of interest in this work is a group norm defined below. 
% We say $\Zmbb_{a}$ is written in \emph{balanced residues} when the elements are defined symmetrically about zero. For example, if $a = 2k+1$ is odd, the elements are labeled $-k$ to $k$.
% \dyh{without balance residues, should be good.}
\begin{definition}[Cyclic metric]\label{def:cyc-norms}
    For $v \in \bigoplus_{i=1}^d  \Zmbb_{a_i}$ and nonzero $w \in \R{d}$, define
    $$
    \cycnormonew{v}{w}  \coloneqq \sum_{i=1}^d |w_i|\,\max\!\big\{v_i,\, a_i - v_i\big\}.$$
\end{definition}

% \dyh{
% So I think we want
% \begin{align*}
%     \lVert \sum_i \lvert w_i \rvert\lvert x_j \pmod{a_j} - x_j' \pmod{a_j}\rvert
% \end{align*}can u go back to waht u had?}

For our purposes, a \emph{lattice} is a discrete abelian subgroup of $\R{m}$ given by
\begin{equation*}
    \Lambda = \left\{ \sum_{i=1}^n z_i B_i : z_i \in \Zmbb \right\}
\end{equation*}
where $B_1,\ldots, B_n \in \R{m}$ are linearly independent. See~\cite{cassels1996introduction} for a comprehensive overview. If the $B_i$ also form a basis ($n=m$), the lattice is \emph{full rank}. We also denote $\Lambda$ by $B \Z{n}$, where $B$ has columns $B_i$. When $\ker(B) \neq 0$, i.e., the columns are linearly dependent, $B \Z{n}$ may not form a lattice but will in certain cases, such as if $B$ has rational entries. A square integer matrix $U$ is \emph{unimodular} if $\det U = \pm 1$. Such a matrix always has unimodular inverse, and two full-rank lattices $A \Z{n}$, $B \Z{n}$ are equal iff $A = B U$ for some unimodular $U$. We can think of $U$ as a change of basis matrix for a lattice, changing the fundamental cell while preserving the volume of that cell.

Any integer matrix $A \in \Z{m\times n}$ can be expressed in \emph{Smith normal form} (SNF) $A = URV$, where $U, V$ are unimodular with dimension $m, n$ respectively, and $R = \diag(r_1, \ldots, r_k)$ is a diagonal integer matrix with $r_i \vert r_{i+1}$ for all $i \in [k-1]$. The rightmost $r_i$ may be zero. The SNF can be computed in time polynomial in the dimension and log of the matrix norm~\cite{iliopoulos1989worst}, with improved performance in sparse instances. Given a full-rank lattice $A\Z{m}$ with square matrix $A$, the quotient group $Q = \Z{m}/ A\Z{m}$ is finite and
\begin{equation*}
    \abs{Q} = \abs{\det A}.
\end{equation*}
In fact, the SNF of $A$ fully determines the group structure in the sense that 
\begin{equation*} 
    \Z{m}/A\Z{m} \cong \bigoplus_{j=1}^m \Zmbb_{r_j}
\end{equation*} 
is the invariant factor decomposition. Specifically, $Q$ has independent generators which are the columns $u_j$ of $U$ (as cosets), where each $u_j$ generates the cyclic group $\Zmbb_{r_j}$.

\subsection{Basics of finite Markov chains}\label{subsec:markov}
Let $\Xcal$ be a finite set of size $N$. A \textit{Markov chain} on $\Xcal$ is a random process that moves between elements of $\Xcal$ according to an $N \times N$ transition matrix $P$, where $P(x, y)$ is the probability of moving from $x$ to $y$. The matrix $P$ is stochastic: $\sum_{y \in \Xcal} P(x, y) = 1$ for all $x \in \Xcal$. Probability distributions $\rho$ on $\Xcal$ evolve according to $\rho^{(t+1)} = \rho^{(t)} P$. After $t$ steps the state is $\rho^{(t)} = \rho^{(0)} P^{\top}$, $\rho^{(0)} = \rho$. A distribution $\pi$ is \textit{stationary} if $\pi P = \pi$. For any function $f: \Xcal \to \mathbb{R}$,
$$
(Pf)(x) = \sum_{y \in \Xcal} P(x, y) f(y).
$$ 

The \emph{reversed Markov chain} $\mathcal{M}^\star = (\Xcal, P^\star, \pi)$ shares the stationary distribution $\pi$ and has transition matrix $P^\star$ defined by
$$
\pi(x) P(x, y) = \pi(y) P^\star(y, x).
$$
The chain is \emph{reversible} if $P^\star = P$, and the chain is \emph{symmetric} if $P = P^{\top}$.

For reversible chains, there exists an inner product space over which $P$ is diagonalizable, with eigenvalues
$$
1 = \lambda_1 > \lambda_2 \geq \cdots \geq \lambda_n \geq -1,
$$
where $\lambda_1 = 1$ and $\lambda_2 < 1$. The \textit{spectral gap} is
$$
\delta = 1 - \max\{ |\lambda| : \lambda \in \{\lambda_2, \dots, \lambda_n\}, \lambda \neq \pm 1 \}.
$$
The chain is \emph{aperiodic} if the set
$\mathcal{T}(x) = \{ t \geq 1 : P^{\top}(x, x) > 0\}$ satisfies $\gcd(\mathcal{T}(x)) = 1, \forall x \in \Xcal$. This can always be assured by making the chain \emph{lazy}: 
\begin{align*}
    P_{\text{lazy}} = \frac{I + P}{2},
\end{align*}
which only halves the spectral gap but preserves the stationary distribution.

The \textit{mixing time} of a Markov chain is the amount of time it takes for the distance to stationarity to be small:
$$ t_{\text{mix}} (\varepsilon) \coloneqq \min \left\{ t : d(t) \leq \varepsilon \right\},$$
where $d(t) \coloneqq \sup_{\mu} \textup{TV}(\mu P^{\top}, \pi)$, and $\textup{TV}$ is the \emph{total variation distance}
$$\textup{TV}(\mu , \nu) \coloneqq \frac{1}{2} \| \mu - \nu \|_{1}.$$
The relationship between mixing time and the spectral gap for a reversible Markov chain can be expressed as 
$$ t_{\text{mix}} (\varepsilon) =\Ocal \left( \delta^{-1} \cdot \polylog \frac{1}{\varepsilon} \right).$$

We also introduce the concept of log-Sobolev inequalities.
\begin{definition}[Log-Sobolev inequality]
\label{defn:log-sobolev}
A Markov chain $\mathcal{M} = (\Xcal, P, \pi)$ satisfies a log-Sobolev  inequality if for all $f : \Xcal \rightarrow \mathbb{R}$,
$$
\frac{1}{2} \mathbb{E}_{x \sim \pi}\expP\left[ (f(x) - f(y))^2\right] \geq \omega\left(\mathbb{E}_{x\sim\pi}[f(x) \ln f(x)] - \mathbb{E}_{x \sim \pi}[f(x)] \ln \mathbb{E}_{x \sim \pi}[f(x)]\right).
$$
The largest such $\omega$ is called the log-Sobolev constant, and is related to the spectral gap via $\omega \leq \delta$.
\end{definition}

The \textit{discriminant matrix} is a useful operator for analyzing random walks:
\begin{definition}[Discriminant matrix]
\label{defn:discriminant}
For $\mathcal{M} = (\Xcal, P, \pi)$, the discriminant matrix is
$$
D(P)_{ji} = \sqrt{P_{ij} P^\star_{ji}} = \left( \operatorname{diag}(\pi)^{1/2} P \operatorname{diag}(\pi)^{-1/2} \right)_{ij}.
$$
\end{definition}

Lastly, we also define the concept of pseudo-Lipschitzness.
\begin{definition}[$P$-pseudo Lipschitz norm] %-- Definition 2.6 \cite{chakrabarti2024generalized}]
\label{defn:pseudo-Lipschitz}
Let $\mathcal{M} = (\Xcal, P, \pi)$ be a Markov chain. The $P$-pseudo Lipschitz norm of $f: \Xcal \rightarrow \mathbb{R}$ is defined to be
$$
   \lVert f \rVert_{P} \coloneqq \max_{x\in \mathcal{X}}\expP[(f(x) - f(y))^2].
$$
\end{definition}

\subsection{Background on linear programming}\label{ss:LP}
The \textit{linear optimization} problem and its dual can be written in standard form as
\begin{equation*}
   \min_{x \in \mathbb{R}^n} \left\{ c^{\top} x : Ax = b,\, x \geq 0\right\}, \quad  \max_{y \in \mathbb{R}^m} \left\{ b^{\top} y : A^{\top} y \leq c\right\},
\end{equation*}
where $A \in \Zmbb^{m\times n}$ and $b \in \Zmbb^m$ define a collection of $m$ linear constraints. Here and throughout, we assume $A$ is full rank. For our purposes, a \emph{basis} $\Bcal \subseteq [n]$ is a subset of $m$ indices such that the vectors $\{A_j\}_{j\in\Bcal}$ form a basis over $\Rmbb^m$ in the usual sense. A vector $x \in \mathbb{R}^n$ is called \emph{feasible} if $x$ satisfies the constraints, i.e., $Ax = b$ and $x \geq 0$. Moreover, a \textit{basic feasible solution} $x = (x_\Bcal, x_{\Ncal})$  is a feasible solution such that its nonzero entries are indexed by $\Bcal$. Note that the values of $x_{\Bcal}$ are uniquely determined by the system $A_{\Bcal } x_{\Bcal } = b$. The associated basic dual solution is $y_{\Bcal} = (A^{\top}_{\Bcal})^{-1} c_{\Bcal}$. 

Let $\mathscr{B}$ denote the collection of all feasible bases of the LP relaxation. For any basis $\Bcal  \in \mathscr{B}$, we define the complement $\mathcal{N} \coloneqq [n] \setminus \Bcal $, and the complementary parts $A_{\Ncal}, c_\Ncal$ of the constraints and the objective function, respectively. We will assume that $m = \Theta(n)$. Associated to any basis $\Bcal $, the \textit{reduced costs} (or dual slacks) are given by $\bar{c} \coloneqq c - A^{\top}_{\Ncal} y_{\Bcal}$. The reduced cost $\bar{c}_j$ for $j \in [n]$ quantifies the rate of improvement in the objective function per unit increase in $x_j$, subject to feasibility.

Let $(x^{\star}, y^{\star}) \in \R{n}_{\geq 0} \times \R{m}$ denote a primal-dual solution corresponding to basis $\Bcal  \in \mathscr{B}$. We say that $\Bcal $ is an \textit{optimal} basis if, in addition to primal-dual feasibility,  \textit{complementarity slackness} is satisfied:
$$
x^{\star} \odot \left(c-A^{\top} y^{\star} \right) = \mathbf{0}_n
$$
where $\odot$ denotes the element-wise product. When the LP is nondegenerate, $\Bcal = \supp(x^\star)$ and $\Ncal = \{i \in [n] : x_i^\star = 0\}$.

When the problem data $(A,b,c)$ is integer, we let $L$ denote the binary input length:
\begin{equation*}
    L \coloneqq \sum_{i=1}^m \sum_{j=1}^n \left[ 1 + \left\lfloor \log_2 |A_{ij}| \right\rfloor \right] + \sum_{i=1}^m \left[ 1 + \left\lfloor \log_2 |b_i| \right\rfloor \right] + \sum_{j=1}^n\left[ 1 + \left\lfloor \log_2 |c_j| \right\rfloor \right].
\end{equation*}
The optimal partition $(\Bcal ,\mathcal{N})$ can be found in time $\Ocal \left( n^3 L\right)$ using interior point methods~\cite{wright1997primal}.

\subsection{The generalized short path algorithm} 

Short path (SP) algorithms, first proposed by Hastings~\cite{hastings2018shortPath} and later refined by Dalzell et al.~\cite{dalzell2022mind}, represent a methodological advancement over the QAA~\cite{farhi2000quantum}. Below, we outline the core concepts.

While quantum computers cannot efficiently prepare the ground state for general Hamiltonians, certain Hamiltonians have ground states that are easy to construct. The QAA leverages this by starting with an initial Hamiltonian $H_0$ whose ground state is simple to prepare—commonly, the transverse-field operator $H_0 = -\sum_{i \in [n]} \sigma^x_i$, where $\sigma^x_i$ is the Pauli $X$ operator on qubit $i \in [n]$. The algorithm then defines a time-dependent Hamiltonian
$$
H_{\text{QAA}} (t) \coloneqq (1-\tau (t)) H_{0} + \tau (t) H,
$$
where $\tau :[0, T] \to [0,1]$ is a smooth, monotonic function satisfying $\tau(0)=0$ and $\tau(T)=1$ so that $H_{\text{QAA}}(0)=H_0$ and $H_{\text{QAA}}(T)=H$. According to the \textit{adiabatic theorem}~\cite{messiah2014quantum}, simulating the Schrödinger equation
$$
i \frac{\text{d}}{\text{dt}} \ket{x(t)} = H_{\text{QAA}}(t)\ket{x(t)}
$$
for sufficiently large $T$ prepares the ground state of $H$, provided that the minimum gap $\delta_\mathrm{min}$ between the lowest two eigenvalues is nonzero. Here
$$
\delta_\mathrm{min} := \min_{t \in [0, T]} \{ \lambda_{2} (H_{\text{QAA}}(t)) - \lambda_{1} (H_{\text{QAA}}(t)) \},
$$
where $\lambda_{1}$ and $\lambda_{2}$ are the smallest and second smallest eigenvalues, respectively.

The challenge is that $\delta_\mathrm{min}$ tends to shrink rapidly with system size, necessitating prohibitive simulation times $T$. Roughly, $T$ scales polynomially with the inverse gap. Obtaining analytical bounds on the spectral gap is notoriously difficult~\cite{hislop2012introduction}, but more importantly, the QAA is hampered by ``localization'' phenomena—quantum analogues of a classical particle getting trapped in a local minimum. %—known as avoided level crossings. 
Localization can induce (super) exponentially small gaps, resulting in run times that may be exponentially slower than classical brute-force search~\cite{altshuler2009adiabatic}.

SP algorithms avoid first-order phase transitions via two key modifications:
\begin{enumerate}
    \item Instead of using the standard adiabatic Hamiltonian $H_{\text{QAA}}$, the algorithm uses a slightly different interpolation between the mixer Hamiltonian $H_0$ and the final Hamiltonian $H$ which encodes a truncated and normalized cost function.
    \item The algorithm does not follow the adiabatic path, but performs two successive jumps called short jump and long jump to avoid the points where the spectral gap of the Hamiltonian decays (super) exponentially.
\end{enumerate}

The \emph{generalized} quantum SP algorithm, as developed in~\cite{chakrabarti2024generalized}, further provides a principled framework for achieving super-quadratic quantum speedups over not just over unstructured search, but also over classical Markov Chain search in combinatorial optimization. The key ingredient is a generalized mixer Hamiltonian over the Markov Chain graph of interest. Below, we review the algorithm and conditions for super-quadratic speedups.

\subsubsection{Short path Hamiltonian and algorithmic steps}
Let $\Mcal = (\Xcal, P, \pi)$ be a reversible, aperiodic Markov chain over a finite set $\Xcal$, where $\pi$ is the stationary distribution of $P$. 
The generalized SP Hamiltonian $H_\mu$ is defined as in \eqref{eq:SP-hamiltonian}. This construction generalizes the original SP framework to arbitrary reversible Markov chains and cost functions, as formalized in~\cite{chakrabarti2024generalized}. We define $E^{\star} \coloneqq \min_{z \in \Xcal} H(z)$, and subtract a scalar from the cost function to ensure $E^{\star} < 0$.

We outline the generalized SP algorithm in Algorithm~\ref{alg:generalized-SP}. The algorithm begins by preparing the ground state $\ket{\sqrt{\pi}} = \sum_{z \in \Xcal} \sqrt{\pi(z)} |z\rangle$ of $-D(P)$, a weighted superposition over the feasible set $\Xcal$. Next, quantum phase estimation and amplitude amplification are employed to carry out the short jump. This enables us to prepare the ground state $\ket{\psi_{\mu}}$ of $H_{\mu}$ for a suitably chosen value of $\mu$. Following this, amplitude amplification is applied to project the state onto the optimal solution subspace, taking advantage of the overlap created by the short jump, leading to the long jump. Finally, a measurement is performed in the computational basis to obtain a feasible solution, which is optimal with high probability. Each step is implemented using block-encoded quantum circuits~\cite{chakraborty2018power} and fixed-point amplitude amplification~\cite{yoder2014fixed}, ensuring polynomial overhead in the relevant parameters.

\begin{algorithm}
  \caption{\textsc{GeneralizedShortPathAlgorithm}}
  \label{alg:generalized-SP}
  \hspace*{\algorithmicindent} \textbf{Input}: Algorithmic parameters $\mu,\eta$. Problem parameters $H, P, \pi, E^{\star}$. \\
 \hspace*{\algorithmicindent} \textbf{Output}: A global minimizer of $H$.
  \begin{algorithmic}[1]
    \State Prepare $\ket{\sqrt{\pi}}$, the ground state of $-D(P)$.
    \State \textbf{Short Jump:} Given block-encodings of $-D(P)$ and $H_{\mu}$, prepare $\ket{\psi_\mu}$ (the ground state of $H_{\mu}$) up to exponentially small error using \cite[Proposition 1]{dalzell2022mind}.
    \State \textbf{Long Jump:} Given block-encodings of $H_{\mu}$ and $H$, prepare $ \frac{\Pi^{\star}\ket{\psi_\mu}}{\|\Pi^{\star}\ket{\psi_\mu} \|}$ up to exponentially small error using procedure outlined in \cite[Proposition 1]{dalzell2022mind}, where $\Pi^{\star}$ orthogonally projects onto the space of  global minimizers in $\Xcal$.
    \State Measure the resulting state in the computational basis.
  \end{algorithmic}
\end{algorithm}
% The long jump step could also be performed using generalized quantum minimum finding \cite[Theorem 49]{van2020quantum}.

\subsubsection{Runtime of the short path algorithm}
\label{sss:runtime_short_path}

The generalized short path algorithm makes use of the following input assumptions:
\begin{assumption}
\label{asm:input-assumptions}
The following subroutines are used as primitives in the Generalized Quantum SP framework.
\label{asm:quantum-input}
    \begin{enumerate}
        \item \textbf{Initial State Preparation:} We assume the existence of a unitary $U_\pi$ implementable using $\poly(n)$ gates, such that $U_\pi|0\rangle = |\sqrt{\pi}\rangle \coloneqq \sum_{z \in \Xcal} \sqrt{\pi(z)} |z\rangle$.
        \item \textbf{Block-encoding of Markov Chain:} We assume the existence of a unitary $U_{D(P)}$, implementable with $\textup{poly}(n)$ gates, that is a $(\kappa_1,a)$-block-encoding of $D(P)$ for $\kappa_1 =\mathcal{O}(\poly(n))$. 
        \item \textbf{Block-encoding of Cost Function:} We assume the existence of a unitary $U_{H}$, implementable with $\textup{poly}(n)$ gates, that is a $(\kappa_2,a)$-block-encoding of $H$ for $\kappa_2~=~\mathcal{O}(\poly(n))$.
    \end{enumerate}
\end{assumption}
\noindent Each of the input assumptions are satisfied in the settings we consider in this paper. We provide explicit procedures for initial state preparation and constructing the necessary block-encodings in Section~\ref{subsec:state_prep_and_block_encoding}. 

Our analysis of the short path framework requires two additional technical assumptions called $\gamma$-\emph{spectral density} and $\Delta_P$-\emph{stability}. The first of these concerns properties of the spectral density near the ground state energy.
\begin{definition}[$\gamma$-Spectral Density -- Definition 3.2 \cite{chakrabarti2024generalized}]
\label{defn:gamma-spectral}
The cost Hamiltonian $H$, with ground state energy $E^{\star} < 0$, is said to satisfy the $\gamma$-spectral density condition with respect to the stationary distribution $\pi$ if:
$$
    \pi(E \leq (1-\eta)E^{\star}) \leq \pi(E^{\star})^{\gamma},
$$
for some $\eta \in (0, 1)$ and $\gamma > 0$.
\end{definition}
\noindent To show a true super-quadratic speedup, we require $\gamma = \Theta(1)$. In this regime, the \emph{near‑optimal sublevel set} $\{x \in \Xcal: \bra{x} H \ket{x} \le (1-\eta)E^\star\}$ does not carry measure that grows super‑polynomially relative to the set of exact optima. If, for every $\eta\in(0,1)$, no constant $\gamma$ exists (i.e., the near‑optimal sublevel set is disproportionately large), then simple uniform sampling or rapidly mixing chains would find $(1-\eta)$‑optimal solutions in subexponential time, yielding a subexponential approximation scheme. Thus, for ``hard'' instances where no such scheme is expected, modeling with a \emph{constant} $\gamma$ is reasonable: the near‑optimal sublevel set is not exponentially larger (in measure) than the optimal set.

Meanwhile, $\Delta_P$-stability is a certain ``smoothness'' condition that we now define.
\begin{definition}[$\Delta_P$-stability -- Definition 3.1 \cite{chakrabarti2024generalized}]
\label{def:delta-stability}
    Let $\Mcal = (\Xcal, P, \pi)$ be a Markov chain and define $\widetilde{\theta}_\eta(x) \coloneqq \theta_\eta \left( \frac{x}{|E^\star|} \right)$. We say that $H$ is $\Delta_P (\eta)$-stable under $\Mcal$ if 
    \begin{align*}
        \expP[\widetilde{\theta}_\eta ( H(y) )] \leq \widetilde{\theta}_\eta ( H(x) + \Delta_P(\eta) ), \quad \forall x \in \Xcal.
    \end{align*} 
\end{definition}
\noindent From \cite[Lemma 4.26]{chakrabarti2024generalized}, we have $\sqrt{\lVert H\rVert_{P}} \geq \Delta_{P}(\eta), \forall \eta \in (0, 1)$. When the value of $\eta$ is not important, we  simply write $\Delta_P$. The $\Delta_P$-stability condition reflects the underlying structure of the problem, that is when $\Delta_P$ is large, the problem effectively becomes an unstructured search, and the achievable speedup is typically that of Grover's algorithm. In fact, the advantage over Grover's algorithm $\alpha(n)$ in Theorem \ref{thm:main-speedup} is inversely proportional to $\Delta_P$.

In Section \ref{sec:primal}, we will show that the null space random walk we design satisfies all of the aforementioned assumptions for the family of group relaxations considered in this work. When the Markov chain $P$ satisfies a Log-Sobolev inequality (Definition \ref{defn:log-sobolev}), we have the following bound on the runtime of the generalized SP algorithm.
\begin{theorem}[{Log-Sobolev Version~\cite[Theorem 4.2]{chakrabarti2024generalized}}]\label{thm:gen_short_path_runtime_ls}
Let $\Mcal = (\Xcal, P, \pi)$ be a reversible, aperiodic Markov chain, and let $H:\Xcal \rightarrow \R{}$ be a diagonal, $\Delta_P$-stable Hamiltonian (Definition \ref{def:delta-stability}) with ground state energy $E^{\star} < 0$. Suppose $\Mcal$ has a log-Sobolev constant lower bounded by $\omega$. Then
 $\gamma$-spectral density (Definition \ref{defn:gamma-spectral}) is satisfied for
$$
    \gamma = \frac{\omega(-(1-\eta)E^{\star} - \mathbb{E}_{\pi}[H])^2}{\lVert H \rVert_{P}\ln(1/\pi(E^{\star}))}.
$$ 
Also if $\mu$ satisfies
\begin{align*}
&\mu< \mu^{\star} \coloneqq \frac{2}{3}\gamma \omega\ln\left(\frac{1}{\pi(E^{\star})}\right)
\end{align*}
and \begin{align}
\label{eqn:ell_condition_ls}
\frac{1}{3\gamma \ln(1/\pi(E^{\star}))\sqrt{ \pi( E \leq (1-\eta)E^*)}} - \max\left(\frac{4\ln(\lvert \Xcal\rvert)}{\omega}, \frac{3\lvert E^{\star}\rvert^2(1-\eta)^2}{\Delta_P^2}\right)  \geq 2,
\end{align}
then there exists a short path algorithm that determines the ground state of $H$ over $\Xcal$ with running time
$$
\mathcal{O} \left(\textup{poly}(n)\omega^{-1}[\pi(E^{\star})^{-1}]^{\left(\frac{1}{2}-\frac{\eta(1-\eta)\lvert E^{\star}\rvert \mu}{2\ln(1/\pi(E^{\star}))\Delta_P}\right)}\right).
$$
\end{theorem}

% \begin{remark}
% Note that in \cite{chakrabarti2024generalized}, the $\mathbb{E}_{\pi}[\widetilde{\theta}_{\eta}(H)]$ term does not have the filter function $\widetilde{\theta}$ applied to it, i.e., it is just $\mathbb{E}_{\pi}[H]$. However, it is implicit in their analysis that we can make this swap.
% \end{remark}

\noindent %Note that any upper bound on $\Delta_P$ suffices; for example, one may use $\sqrt{\lVert H\rVert_{P}}$.
As remarked in \cite[Remark 4.4]{chakrabarti2024generalized}, the additional conditions in the above theorem--on top of spectral density, stability and the maximum-allowed $\mu$--are considerably mild since the denominator in the first term of \eqref{eqn:ell_condition_ls} is exponentially small.

There is also a version of the generalized SP runtime that only depends on the spectral gap of the walk which is slightly weaker, limiting its applicability.
\begin{theorem}[{Spectral Gap Version \cite[Theorem 4.4]{chakrabarti2024generalized}}]
\label{thm:gen_short_path_runtime_poincare}
Let $\Mcal = (\Xcal, P, \pi)$ be a reversible, aperiodic Markov chain, and let $H: \Xcal \rightarrow \R{}$ be a diagonal, $\Delta_P$-stable Hamiltonian  with ground state energy $E^{\star} < 0$, $P$-pseudo Lipschitz norm $\lVert H \rVert_{P}$. In addition, suppose $\Mcal$ has a spectral gap lower bounded by $\delta$. Then
 $\gamma$-spectral density is satisfied for
\begin{align*}
    \gamma = \frac{\sqrt{\delta}(-(1-\eta)E^{\star} - \mathbb{E}_{\pi}[\widetilde{\theta}_{\eta}(H)])}{\sqrt{\lVert H \rVert_{P}}\ln(1/\pi(E^{\star}))},
\end{align*}
and suppose that $\frac{\delta}{2} > \pi(E^{\star})^{\gamma}$.
Also if $\mu$ satisfies
\begin{align*}
    \mu < \mu^{\star} :=\frac{\delta}{4}
\end{align*}
and
 \begin{align}
\label{eqn:ell_condition_pc}
\frac{1}{2\sqrt{\pi( E \leq (1-\eta)E^*)}} - \max\left(\frac{4\ln(\lvert \Xcal\rvert)}{\delta}, \frac{3\lvert E^{\star}\rvert^2(1-\eta)^2}{\Delta_P^2}\right)  \geq 2,
\end{align}
then there exists a short path algorithm that determines the ground state of $H$ over $\Xcal$ with running time
$$
\mathcal{O} \left(\textup{poly}(n)\delta^{-1}[\pi(E^{\star})^{-1}]^{\left(\frac{1}{2}-\frac{\eta(1-\eta)\lvert E^{\star}\rvert \mu}{2\ln(1/\pi(E^{\star}))\Delta_P}\right)}\right).
$$
\end{theorem}

The following provides sufficient conditions under which the SP algorithm provides a super-quadratic speedup over classically sampling from $\pi$.
\begin{corollary}[{Super-Quadratic Condition \cite[Corollary 4.5]{chakrabarti2024generalized}}]
\label{cor:super_quad_runtime_short_path}
Let $\Mcal = (\Xcal, P, \pi)$ be a reversible, aperiodic Markov chain, and let $H: \Xcal \rightarrow \R{}$ be a diagonal Hamiltonian with ground state energy $E^{\star} < 0$.
Suppose that the chain $\Mcal$ and cost-function $H$ are such that the conditions of Theorem~\ref{thm:gen_short_path_runtime_ls} or \ref{thm:gen_short_path_runtime_poincare} are satisfied for some $\mu = \Theta(1)$ and 
%$b^{\star}$ result in $b^{\star}$ and $\gamma$ that are independent of $n$. \guney{what does this mean} If  
\begin{align*}
    \frac{\lvert E^{\star}\rvert}{\Delta_{P}} = \Theta\left(\log\left(1/\pi(E^{\star})\right)\right).
\end{align*} 
Then under Assumption~\ref{asm:input-assumptions}, the short path algorithm finds the minimizer of $H$ over $\Xcal$ with running-time bounded by
\begin{align*}
    \mathcal{O}\left([\pi(E^{\star})^{-1}]^{\frac{1}{2} - \alpha } \; \poly(n)\right),
\end{align*}
where
$$
    \alpha = \frac{\eta(1-\eta)\lvert E^{\star}\rvert \mu}{2\Delta_{P}\ln(1/\pi(E^{\star}))} = \Theta(1).
$$ 
\end{corollary}
We note that even though the conditions in Corollary \ref{cor:super_quad_runtime_short_path} seem complicated and somewhat contrived, the speedup factor (say in the setting of  Theorem \ref{thm:gen_short_path_runtime_ls}) simplifies to $\alpha(n)= \Omega\left(\frac{\omega n}{\Delta_P}\right)$ for typical hard combinatorial optimization problems. This because typically $|E^{\star}|=\Theta(n)$ and $\ln(1/\pi(E^{\star})) = \Ocal(n)$ since the spectral density condition holds (otherwise we can find approximate solution in sub-exponential time). As an example, consider the hypercube walk $P$, where $D(P) =\frac{1}{n} \sum_{i \in [n]} \sigma^x_i$. Since $\omega=\Omega(1/n)$ for this walk, $\alpha = \Omega(1)$ as long as the change in the objective function after a bit-flip, bounded by $\Delta_P$ is $\Ocal(1)$ . For general problems, the challenge is to construct a random walk that has large Log-Sobolev constant and small pseudo-Lipschitz norm or $\Delta_{P}$. In the next sections, we describe the Markov chain construction to satisfy these conditions for a certain class of integer programming problems.  

\section{Solving group relaxations through local search}\label{sec:primal}
In this section, we first review important details regarding the group relaxation in order to derive a formulation of this problem amenable to local search. We then show how to solve the reformulated problem using \textit{Markov chain search}, and analyze the complexity of the resulting scheme. 

\subsection{The group relaxation}\label{ss:primal_group_problem}
Let $\Bcal $ denote an optimal basis for the LP relaxation of \eqref{e:ILP}. Recall the definition of the group relaxation induced by this basis 
\begin{equation}\tag{$\text{ILP}_{\Bcal }$}\label{e:group_problem_recall}
       \OPT_\Bcal \coloneqq  \min \left\{c^{\top} x: x \in \Pcal_{\Bcal} \right\},
\end{equation}
where $\Pcal_{\Bcal} := \{ (x_{\Bcal} , x_{\Ncal}) \in \Z{m} \times \Z{n-m}_{\geq 0} : Ax = b \}$.  

Since any $x = (x_{\Bcal}, x_{\Ncal}) \in \Pcal_{\Bcal}$ satisfies $Ax = b$, we may write 
\begin{align*}
    c^\top x = c_\Bcal ^\top x_\Bcal + c_\mathcal{N}^\top x_\mathcal{N} &= c_\Bcal ^\top A_\Bcal ^{-1}(b-A_\mathcal{N}x_\mathcal{N})+c_\mathcal{N}^\top x_\mathcal{N}
    = c_\Bcal ^\top A_\Bcal ^{-1}b + \bar{c}_{\Ncal}^{\top}x_\mathcal{N}.
\end{align*}
Given that $A_{\Bcal}$ is invertible, we can also reformulate the system $A_{\Bcal} x_{\Bcal} + A_{\Ncal} x_{\Ncal} = b$ as
\begin{align*}
    A^{-1}_{\Bcal}A_{\Ncal}x_{\Ncal} + x_{\Bcal} =  A^{-1}_{\Bcal}b \implies x_{\Bcal} =  A^{-1}_{\Bcal}(b-A_{\Ncal}x_{\Ncal}). 
\end{align*}
Thus, once $x_{\Ncal}\in \Z{n-m}_{\geq 0}$ is fixed, $x_{\Bcal} \in \Z{m}$ is uniquely determined. Accordingly, \eqref{e:group_problem_recall} is equivalent to
\begin{equation} \label{eq:basis_OPT}
\OPT_\Bcal =  \min_{x_{\mathcal{N}} \in \Zmbb _{\geq 0}^{n-m}} \left\{ c_{\Bcal }^{\top} A_{\Bcal }^{-1} b + \bar{c}_{\mathcal{N}}^{\top} x_{\mathcal{N}} : A_{\Bcal}^{-1}A_{\Ncal}x_{\Ncal}  = A_{\Bcal}^{-1}b \pmod{\Z{m}}\right\}.
\end{equation} 
Since $\Bcal$ is optimal for \eqref{e:LP}, the nonbasic part of the reduced cost is nonnegative, 
\begin{align*}
    \bar{c}_{\Ncal}^\top \coloneqq c_\mathcal{N}^\top-c_\Bcal ^\top A_\Bcal ^{-1}A_\mathcal{N} \geq 0,
\end{align*}
and $c_{\Bcal }^{\top} A_{\Bcal }^{-1} b = \OPT_{\text{LP}}$, so we have $\OPT_\Bcal \geq \OPT_{\text{LP}}$.

To further refine this formulation, consider the Smith normal form $A_{\Bcal} = URV$, where $U$ and $V$ are unimodular matrices, and $R = \textup{diag}(r_1, \dots, r_m)$ with 
\begin{equation*}
    \abs{\det A_\Bcal} = \prod_{j=1}^{m} r_j, \quad r_1~|~r_2~|~\cdots~|~r_m.
\end{equation*}
Each $r_j > 0$, as $A_\Bcal$ is full-rank. Multiplying both sides of the constraint in \eqref{eq:basis_OPT} by $RV$, we obtain
\begin{equation*} 
    \OPT_\Bcal =   \OPT_{\text{LP}}  + \min_{x_{\mathcal{N}} \in \Zmbb _{\geq 0}^{n-m}} \left\{  \bar{c}_{\mathcal{N}}^{\top} x_{\mathcal{N}} : U^{-1}A_{\Ncal}x_{\Ncal}  = U^{-1}b \pmod{R\Z{m}}\right\}.
\end{equation*}
%where the notation mod $R \Z{}$ signifies that the $i$-th row of the congruence is considered modulo $r_i$, and we define $G \coloneqq \lat /\lat_{\Bcal}  \cong R\Zmbb $. 
Labeling the data $\widetilde{\Ambf}  = U^{-1}A_{\Ncal} \in \Zmbb ^{m \times (n-m)}$, $\bmbf  = U^{-1}b \in \Z{m}$, and $\mathbf{c} = \bar{c}_{\mathcal{N}}$, we arrive at
\begin{equation}\label{e:GILO_infinite}
   \OPT_\Bcal =   \OPT_{\text{LP}}  +\min_{\xmbf \in \Zmbb _{\geq 0}^{n-m}} \left\{ \mathbf{c}^{\top} \xmbf : \widetilde{\Ambf} \xmbf  = \bmbf  \pmod{R\Z{m}}\right\}.
\end{equation}
Note it may be the case that some columns of $\widetilde{\Ambf }$ are already in $R\Z{m}$. For those columns, we can simply set $x_j = 0$, as this process never increases the objective value, nor changes the LHS of the constraints modulo $R \Z{m}$. Denoting by $\Ambf  \in \mathbb{Z}^{m \times d}$, with $d \leq n-m$, the matrix obtained upon dropping these columns, problem~\eqref{e:GILO_infinite} simplifies to
\begin{equation}\tag{$\widetilde{\text{ILP}}_{\Bcal }$}\label{e:GILO-final}
      \OPT_\Bcal =   \OPT_{\text{LP}} +  \min_{\xmbf \in \Zmbb _{\geq 0}^{d}} \left\{\mathbf{c}^{\top} \xmbf : \Ambf  \xmbf = \bmbf  \pmod{R\Z{m}} \right\}.
\end{equation}
We thus describe an instance of a group relaxation via the tuple $(\Ambf , \bmbf , \mathbf{c}, \{r_j\}_{j=1}^{m})$. 

Gomory originally showed that, under certain conditions, the group relaxation is \emph{asymptotically tight}---that is, it solves the original ILP exactly when the right-hand side is sufficiently large. 
\begin{theorem}[Theorem 5.6 in \cite{wolsey1999integer}]
    Let $\Bcal$ be any feasible basis for the \ref{e:LP} relaxation of the \ref{e:ILP} problem. The group relaxation defined by $\Bcal$ solves the \ref{e:ILP} problem to optimality if 
    $$ A^{-1}_{\Bcal } b \geq \max_{i,j} \left\{ | (A^{-1}_{\Bcal } A_{\mathcal{N}})_{ij} | \right\} \cdot | \det A_{\Bcal } | \mathbf{1}_{n-m}.$$ 
\end{theorem}
Lasserre later established a \emph{nondegeneracy condition}: 
\begin{definition}[Nondegeneracy condition for group relaxations \cite{lasserre2004generating}]
\label{def:non_degen}
    Denote the set of all optimal bases to the LP relaxation of \eqref{e:ILP} by $\mathscr{B}$. Let $\Bcal \in \mathscr{B}$ be an optimal LP basis. If there exists another basis $\Bcal^{\prime} \in \mathscr{B}$ such that $\OPT_\Bcal = \OPT_{\Bcal^{\prime}}$ then the problem is degenerate. Otherwise, the group relaxation defined by $\Bcal$ is nondegenerate, and solves \eqref{e:ILP} to optimality.
\end{definition}
\noindent This nondegeneracy condition can be viewed as a similar statement to the usual notion of nondegeneracy in linear programming, in the following sense: the situation in which the group relaxation only provides a lower bound on $\OPT$ arises, necessarily, when there is another LP basis $\Bcal'$ whose group relaxation yields the same lower bound. We emphasize that group nondegeneracy is strictly stronger than LP nondegeneracy. Moreover, while LP nondegeneracy can often be certified a posteriori (e.g., via reduced costs and basic variable signs), no obvious efficient a posteriori certificate is available for the group case. That said, there is a simple \emph{sufficient} a posteriori check for exactness: if a group‑optimal solution is ILP‑feasible (i.e., satisfies the dropped nonnegativity), then $\OPT_{\Bcal}=\OPT$; conversely, failure of this check does not certify degeneracy, since there may be another basis which induces a group relaxation that solves the ILP.

\subsection{Null space formulation}
\label{subsec:nullspaceform}
Recall from \eqref{e:GILO-final} of the prior section that the feasible set to the group relaxation is 
\begin{align*}
    \{ \xmbf  \in \mathbb{Z}_{\geq 0 }^{d} : \Ambf \xmbf  = \mathbf{b} \pmod{R\Z{m}}\}.
\end{align*}
This set is infinite, provided it is nonempty. However, note that the constraints are linear equations over the finite group $\Rcal \coloneqq \Z{m}/R\Z{m} \cong \bigoplus_{i=1}^m \Zmbb_{r_i}$. Thus, the order of $\Ambf_j$ is finite, and necessarily divides $r_m$. Reducing each component modulo $r_m$ never increases the objective value $\cmbf^\top \xmbf$, since $\cmbf \geq 0$. We can thus take the feasible set to be a subset of $Z \coloneqq \Zmbb_{r_m}^d$ while preserving optimality and feasibility. We note that the ambient space of $\xmbf$ can be even further reduced with additional refinements, and explore this approach in Appendix \ref{app:kernel_compression}.

Suppose $\Ambf \xmbf  = \mathbf{b} \pmod{R\Z{m}}$ is solvable, with some solution $\hat{\xmbf} \in Z$. Then, the set of solutions can be written as the coset $\hat{\xmbf} + K$. Here, $K$ is the kernel of $\Ambf$ as a homomorphism from $Z$ to $\Rcal$
\begin{align}
\label{eq:kernel_set_defn}
    K = \{ \xmbf \in Z : \Ambf \xmbf = \mathbf{0} \pmod{R\Z{m}}\}.
\end{align}
Thus, as typical with more familiar linear systems over vector spaces, one can parameterize the feasible set of the group relaxation using any feasible solution together with the kernel of some group homomorphism. We shall see, however, that characterizing the kernel is somewhat more challenging than the vector space setting.

Supposing the group problem is feasible, we can formulate the objective function in terms of the kernel subgroup as follows
\begin{equation}
\label{eq:transformed_cost_no_basis}
    \min_{\xmbf \in \hat{\xmbf}+K} \mathbf{c}^{\top} \xmbf + c_\Bcal^{\top} A_\Bcal^{-1} b   = \min_{\xmbf \in K}\cmbf^\top (\hat{\xmbf} + \xmbf \pmod{r_m \Z{d}}) + c_{\Bcal }^{\top} A_\Bcal^{-1} b  =: \min_{\xmbf \in K} f(\xmbf).
\end{equation}
The constant shift $c_{\Bcal }^{\top} A_{\Bcal }^{-1} b$ ensures that $\min_{\xmbf \in K} f(\xmbf) = \OPT_\Bcal$.

\subsection{Computing the null space}
\label{sec:computing_the_nullspace}
Our formulation \eqref{eq:transformed_cost_no_basis} is only useful algorithmically provided there are efficient methods for 
\begin{enumerate}[label=(\alph*)]
    \item producing a feasible solution $\hat{\xmbf} \in Z$ to $\Ambf \xmbf  = \mathbf{b} \pmod{R\Z{m}}$; and 
    \item accessing elements $\xmbf \in K$.
\end{enumerate}
 We shall show how to accomplish both of these tasks in classical polynomial time using tools from computational number theory \cite{cohen2000advanced}. We remark that one can formulate \eqref{eq:transformed_cost_no_basis} as an instance of the abelian Hidden Subgroup Problem by considering the task of finding the kernel of $K$ as a hidden subgroup of $Z$. This is a problem for which quantum computers are known to be effective. However, because $\Ambf $ is stored in classical memory rather than provided via an oracle, the setting is somewhat distinct, and in any case, an efficient classical algorithm which produces a set of generators is preferable.

Solving  $\Ambf \xmbf  = \mathbf{b} \pmod{R\Z{m}}$ is complicated by the fact that $\Rcal$ is rectangular, as opposed to square: we would rather work in a space of the form $a \Z{m}$ for some positive integer $a$ because it remains invariant under unimodular transformations. This can be accomplished with some preconditioning. Note that $\xmbf \in Z$ satisfies $\Ambf \xmbf  = \bmbf \pmod{R\Z{m}}$, if and only if
\begin{align*}
    B\Ambf \xmbf  = B\bmbf \pmod{r_m\Z{m}},
\end{align*}
where ${B} =\textup{diag}(r_m/r_1, \dots, r_{m}/r_{m-1}, 1)$. We shall presently see how this aides us.

Given $B\Ambf  \in \Zmbb^{m \times d}$, we first compute the Smith normal form $B\Ambf  = U T V$ for $U \in \Z{m\times m}$, $V \in \Z{d\times d}$ unimodular and $T = \textup{diag}(t_1, t_2, \dots, t_k)$ consisting of invariant factors, with $k = \min(m,d)$. Multiplying by $U^{-1}$, the system reduces to 
\begin{equation} \label{eq:half_reduce}
    TV\xmbf = \bmbf' \pmod{r_m \Z{m}},
\end{equation}
where $\bmbf' \coloneqq U^{-1} B\bmbf$. To solve this system, observe that by unimodularity, the operation $\mathbf{y} \mapsto V^{-1}\mathbf{y} \pmod{r_m\Z{d}}$ defines an automorphism on $Z$. Thus, any solution $\xmbf \in Z$ to \eqref{eq:half_reduce} can be obtained uniquely from some $\ymbf \in Z$ that solves
\begin{equation}\label{eq:simple_lin_cong}
    T \ymbf  = \bmbf' \pmod{r_m \Z{m}}.
\end{equation}
The problem has thus been reduced to $m$ independent linear equations modulo $r_m$
$$
t_i \ymbf_i = \bmbf_i' \pmod{r_m}
$$
whose solutions are well known from number theory. A solution exists if and only if
\begin{equation} \label{eq:diagonalized_modulo}
    \gcd(r_m, t_i) \;\big|\; \bmbf'_i
\end{equation}
for each $i \in [k]$, and in this case a particular solution $\hat{\ymbf}_i$ can be found through the extended Euclidean algorithm. Thus, $V^{-1} \hat{\ymbf} \pmod{r_m \Z{m}}$ is a feasible solution in $\Zmbb_{r_m}^d$.

It remains to determine the kernel $K$, and this can be done by setting $\bmbf' = 0$. In this case, the $i$th equation in \eqref{eq:diagonalized_modulo} has $\gcd(r_m, t_i)$ solutions: the multiples of $r_m/\gcd(r_m, t_i)$. The number of solutions is thus $\prod_{i \in [k]} \gcd(r_m, t_i)$. The solutions $\xmbf\in \Zmbb_{r_m}^d$ to \eqref{eq:half_reduce} are obtained as
\begin{equation}
\label{eq:null_solution_set}
    \xmbf = \sum_{i=1}^k n_i h_i, \quad n_i \in \{0,1,\ldots, \gcd(r_m, t_i)-1\}.
\end{equation}
where $h_i = V_i^{-1} \frac{r_m}{\gcd(r_m, t_i)} \pmod{r_m \Z{d}}$ and all addition is entrywise modulo $r_m$. To summarize,
\begin{equation} \label{eq:K_isomorphism}
    K \cong \bigoplus_{i=1}^k \Zmbb_{\gcd(r_m, t_i)}
\end{equation}
with corresponding cyclic generators $\Scal = (h_i)_{i=1}^k$.

The procedure outlined above computes a generating set for $K$, and shows that transforming the problem in \eqref{e:GILO-final} to \eqref{eq:transformed_cost_no_basis} is classically efficient, enabling us to efficiently compute a generating set for $K$. We summarize the overall procedure in Algorithm~\ref{alg:nullspacereform}.

\begin{algorithm}
  \caption{\textsc{NullGenFinding}}
  \label{alg:nullspacereform}
  \hspace*{\algorithmicindent} \textbf{Input}: Constraints from Group relaxation  $(\Ambf , \{r_j\}_{j=1}^{m})$\\   
 \hspace*{\algorithmicindent} \textbf{Output}: An independent generating set $\Scal$ 
 for $K\leq \Zmbb_{r_m}^d$, the solution space to $\Ambf \xmbf = \mathbf{0} \pmod{R\Z{m}}$.
  \begin{algorithmic}[1]
    \State Compute the SNF of $B\Ambf  \in \Z{m \times d}$: $B\Ambf  =  UTV$, $T = \diag(t_1,\dots,t_k)$, ${B} =\diag (r_m/r_j)_{j=1}^m$.
    \State For each $i \in [k]$ compute $\gcd(r_m, t_i)$ and set $z_i = r_m/\gcd(r_m, t_i)$.
    \State Set $h_i = z_i V^{-1}_i \pmod{r_m\Z{d}}$ for $i \in [k]$
    \State Output $\Scal = (h_1, \dots, h_k)$.
  \end{algorithmic}
\end{algorithm}

\noindent Proof of correctness is given by the above discussion. The runtime of Algorithm~\ref{alg:nullspacereform} can be bounded as follows.  
\begin{lemma}[Runtime of \textsc{NullGenFinding}]
\label{lem:kernel_gen_finding}
    Suppose one has data $(\Ambf , \bmbf , \mathbf{c}, \{r_j\}_{j=1}^{m})$ corresponding to the group relaxation of an ILP stored in classical memory. There is a classical algorithm that outputs cyclic generators $\Scal = (h_i)_{i=1}^k$ for the kernel $K$ of $\Ambf$ in $Z$ using time $\widetilde{\mathcal{O}}(\textup{poly}(\log\lVert \Ambf  \rVert_{\infty}, d) )$. In particular, $K \cong \Zmbb _{\lvert \langle h_1\rangle\rvert} \oplus \cdots \oplus \Zmbb _{\lvert \langle h_k\rangle\rvert}$.
\end{lemma}
\begin{proof} 
For dense $\Ambf $, the Smith normal form can be computed using  
$$
    \mathcal{O}(\textup{poly}(\log\lVert \Ambf  \rVert_{\infty}, d))
    $$
operations, where $\| \Ambf  \|_{\infty} \coloneqq \max_{i \in [d], j \in [d]} |\Ambf _{ij}|$, see~\cite{storjohann1996near}. The only other operations involved are multiplication by $V^{-1}$ and modular arithmetic ($V^{-1}$ can be obtained along with $V$ ``on the fly" during the SNF calculation, by composing the intermediate row and column operations). %We then need to determine a generating set for each of the resulting $m$ scalar linear congruences determined from the diagonal matrix $S$ with the other variables being free. This gives a generating set for $T\zmbf = \mathbf{0}\pmod{R\Z{m}}$, $m \leq d$. 
These costs, are subdominant, hence the runtime follows simply. The last statement from theorem follows from $V^{-1}$ being automorphism for $\mathbb{Z}_{r_m}^d$.
\end{proof}

\subsection{Markov-chain search algorithms for the group relaxation}\label{subsec:local_search_nullspace}
% \label{sec:int_prog_and_group_prob}
In this section, we design a random-walk-based approach for sampling optimal solutions of the group relaxation. 
%This will be accomplished by mapping \eqref{e:GILO-final} to a problem over a finite group, while preserving the optimal solution. Then, 
We will show how to construct a random walk over a coset corresponding to the set of feasible solutions, enabling one to run Markov chain search to solve the group relaxation. The motivation for considering Markov-chain search is primarily for analysis: it is easier to demonstrate quantum speedup in this framework compared with a classical algorithm that directly samples uniformly over the feasible coset. 

As done in the previous section, by exploiting the algebraic structure of the group relaxation, one can efficiently generate a basis for the kernel of the linear map defined by the coefficient matrix with the range  mapped to $\Rcal$. Hence, provided any feasible solution to the group relaxation, we can determine the optimal solution by solving a problem over the null space.  

An efficient procedure for obtaining a generating set for $K$ buys us nothing on its own. It does however enable one to construct a Markov Chain that explores the feasible region of the group relaxation, and mixes to a stationary distribution $\pi$ in $\poly(n)$ steps. Letting $\pi^\star$ denote the total probability that a sample from $\pi$ is a global minimizer of $f$, drawing $\Ocal\left((\pi^\star)^{-1} \log(\delta^{-1})\right)$ samples from $\pi$ suffices to ensure that we encounter the global minimizer with probability at least $1 - \delta$. This algorithm, following \cite{chakrabarti2024generalized}, is called the \emph{Markov Chain Search} (MCS). We outline this framework in Algorithm~\ref{alg:MarkovChainSearch}. Under the stated prerequisities, MCS finds the global minimizer of $f$ in time $\mathcal{O}\left((\pi^\star)^{-1}\cdot \poly (n, L)\right)$. The major component is the construction of the random walk over the feasible set, using the generator-finding procedure from the previous section.

A natural choice for our setting is a \emph{Cayley walk} generated by the set $\Scal = \{h_1, \dots, h_k\}$ returned by Algorithm \ref{alg:nullspacereform}, which leads to a walk over an undirected graph. In particular, a step of the walk corresponds to uniformly at random picking a generator $h_j \in \Scal$, then uniformly at random adding $h_j$ or $-h_j$ to the current point modulo $r_m\mathbb{Z}^d$ producing another element of $K \subset Z$. Let $u_j = |\langle h_j \rangle|$ for $j \in [k]$. Via the isomorphism 
\begin{align}
\label{eqn:hcal_isomorph}
        K \cong \bigoplus_{i \in [k]} \Zmbb_{u_i},
\end{align} 
one can see that this Cayley walk corresponds to a product of cycle walks over the cyclic groups $\Zmbb _{u_1}, \dots, \Zmbb_{u_k}$. %, which effectively corresponds to a Cayley walk over $K$  via the isomorphism in \eqref{eqn:hcal_isomorph}.  
Thus, the stationary distribution $\pi$ is the uniform superposition over $K$, and satisfies $\pi^\star = \lvert K^{\star} \rvert/\lvert K\rvert$. The log-Sobolev (LS) constant of the cycle walk on $\Zmbb _{u_i}$ is $\Theta (1/u_i^2 )$ \cite[Example 4.2]{diaconis1996logarithmic}, and the product of $k$ cycles has a LS constant lower bounded by $\Omega ( 1/(k \min_{i \in [k]} u_i^2))$ \cite[Lemma 3.2]{diaconis1996logarithmic}. This also happens to be the same order as the spectral gap.

\begin{algorithm}
  \caption{\textsc{GroupMarkovChainSearch}}
  \label{alg:MarkovChainSearch}
  \hspace*{\algorithmicindent} 
  \textbf{Prerequisites:} Group relaxation data $(\Ambf , \bmbf , \mathbf{c}, \{r_j\}_{j=1}^{m})$, distribution $\pi$ such that $\pi({\OPT}_\Bcal)$ is the probability that a sample from $\pi$ is a global minimizer, a Markov Chain with transition matrix $P$ that preserves $\Ambf \xmbf = \bmbf  \bmod{R\Z{m}}$, mixing time $\tmix(\pi^\star/2) = \mathcal{O}(\poly(n))$.  \\ 
  \hspace*{\algorithmicindent} \textbf{Input}: Description of $\poly(n)$ time procedures to evaluate $f$ and perform a step of the Markov Chain described by $P$, failure probability $\epsilon$. Known feasible point $\hat{\mathbf{x}} \in \overline{\Pcal}_{\Bcal}$. \\   
 \hspace*{\algorithmicindent} 
 \textbf{Output}: A global minimizer $\xmbf ^{\star}$ of the group relaxation.
  \begin{algorithmic}[1]
    \State Set $i=0$, set $z^{(0)} = \hat{\mathbf{x}}$, and construct $f(\mathbf{x})$ via Equation \eqref{eq:transformed_cost_no_basis}.
    \While {$i \leq \frac{2}{\pi({\OPT}_\Bcal)}\log\left(1/\epsilon\right)$}
        \State Simulate $\tmix$ steps of $P$ to obtain sample $\tilde{z}$.
        \If {$f(\tilde{z}) \le f(z^{(i)})$}
        \State Set $z^{(i+1)} = \tilde{z}$ 
        \Else 
        \State Set $z^{(i+1)} = z^{(i)}$ 
        \EndIf
        \State $i \gets i + 1$
    \EndWhile
    \State Output $z^{(i)}$.
  \end{algorithmic}
\end{algorithm}

Given the above facts, the next result follows immediately. 
\begin{theorem}[Markov chain search for group relaxations] 
\label{thm:mcs_complexity}
Algorithm~\ref{alg:MarkovChainSearch}, with a Cayley walk over $K$ as the input Markov chain, finds an optimal solution to the group relaxation with high probability, and runtime at most $\mathcal{O}\left((|K|/\lvert K^{\star}\rvert) \cdot \poly (n, L)\right)$. 
\end{theorem}
% \noindent As mentioned in Section \ref{subsec:nullspaceform}, the procedure \textsc{NullGenFinding} outputs a generating set $\Scal = \{h_1, \dots, h_k\}$ for $K$, where each vector $a \in \bigoplus_{i \in [k]} \Zmbb_{u_i}$ maps to a unique element of $K$ via $a \mapsto \sum_j a_j h_j$. Hence by randomly sampling from $\bigoplus_{i \in [k]} \Zmbb_{u_i}$ and taking linear combinations in $\Scal$, we can randomly sample from $K$. If $K^{\star} \subset K$ is the subset of minimizers of the group relaxation, then the probability that this uniform sampling finds a minimizer is $\lvert K^{\star} \rvert/\lvert K\rvert$. 

\subsection{Comparison with the shortest path on the Cayley graph}\label{s:compare}

The classical approach for solving \eqref{e:GILO-final} leverages the nonnegativity of $\mathbf{c}$ to reduce the group relaxation to finding the shortest path in a Cayley directed graph (digraph)~\cite{gomory1965relation, shapiro1968dynamic}. Unlike MCS, this algorithm does not explicitly depend on $\lvert K \rvert$. It instead depends on $\lvert G \rvert$, where
\begin{align}
\label{eqn:gdefin}
\begin{aligned}
    G &\coloneqq (\Ambf\Z{d} + R\Z{m})/R\Z{m} \leq \Z{m}/R\Z{m}.
\end{aligned}
\end{align}
In words, the group $G$ is the range of the map $\xmbf  \in \Z{d}_{\geq 0} \mapsto \Ambf \xmbf \pmod{R\Z{m}}$. In this section, we recall the shortest-path approach of Gomory and present a setting under which $\lvert K \rvert = \lvert G \rvert$, making the shortest-path and MCS approaches comparable. Hence, we want to identify settings in which our quantum algorithm has a super-quadratic speedup over Gomory's shortest-path approach.

Let $\Gcal$ be the corresponding Cayley digraph over $G$ generated by \begin{align*}
    \Acal  = \{\Ambf_1 \pmod{R\Z{m}}, \cdots, \Ambf_d \pmod{R\Z{m}}\},
\end{align*} where an edge incident on $\mathbf{y} \in G$ corresponding to $\mathbf{y} \mapsto \left[\mathbf{y}  + \Ambf _{j}\right] \pmod{R\Z{m}}$ has a weight of $\cmbf_j \geq 0$. Define \begin{align*}
    \mathbf{b}' = \left(\bmbf \pmod{R\Z{m}}\right) \in G.
\end{align*} Feasible solutions correspond to directed paths from $\mathbf{0}$ to $\mathbf{b}'$. Since $\mathbf{c}$ is nonnegative, the optimal solution is the feasible solution with smallest $\ell_1$-norm, i.e., the shortest path from $\mathbf{0}$ to $\mathbf{b}'$ in $\mathcal{G}$.  

This approach has complexity dependent on $\lvert G\rvert$, as there are $\mathcal{O}(\lvert \Acal  \rvert \lvert G\rvert)$ edges in $\Gcal$. Thus, shortest-path algorithms can solve the problem with $\mathcal{O}\left(n\lvert G\rvert\right)$ time. In contrast, the null space approach depends on $\lvert K\rvert$. In general, either $\lvert K \rvert$ or $\lvert G \rvert$ may be larger. More precisely, \eqref{eq:kernel_set_defn} and the first isomorphism theorem says that $G \cong \mathbb{Z}_{r_m}^d/K$, and in particular $r_m^d = \abs{K}\abs{G}$. As discussed in Appendix \ref{app:kernel_compression}, choosing  $\bigoplus_{i=1}^d \Zmbb_{s_i}$, where $s_i \coloneqq \abs{\langle\Ambf_i\rangle}$, as the group relaxation domain can result in a smaller kernel $K'$. Then the relation is $\prod_{i=1}^d s_i = \abs{K'}\abs{G}$. As expected, the range and kernel sizes are inversely related for fixed domain, but in general either could be larger than the other. 

In fact, it appears one can construct families of instances where either the range space or kernel space algorithms can be asymptotically faster than the other. The goal of the following result is to show there at least exists a seemingly generic family of ILPs, where $\lvert G \rvert = \lvert K\rvert$.

\begin{theorem}
\label{thm:family_construction}
    There is a constructable family of integer linear programming problems, parameterized by $m\in \Zmbb_+$, whose group relaxations satisfy $G \cong K$ (in particular, $\abs{G} = \abs{K}$).
\end{theorem}
\begin{proof}
Consider arbitrary positive integers $t, m$. Let $n=2m$, $r = t^2$, $R = r I_m $, and $T = t I_m $. Pick unimodular matrices $U, V, U', V' \in \mathbb{Z}^{m \times m}$. We then let the ILP constraint matrix be $A = [URV | U'TV'] = [A_{\Bcal} | A_{\Ncal}] \in \mathbb{Z}^{m \times n}$. Additionally, we require $U, V$ are chosen such that $A_{\Bcal}^{-1}$ has all non-negative entries.

 We want the right-hand side vector $b \in \mathbb{Z}^{m}_{\geq 0}$ to be such that the original ILP is feasible. To do this, we first look at the group relaxation constraint of 
 \begin{align*}
U^{-1}A_{\Ncal}\xmbf  = U^{-1}b  \pmod{r\Z{m}}.
\end{align*}
The feasibility of the group relaxation hinges on 
\begin{align*}
    T\mathbf{z} = b \pmod{r\mathbb{Z}^m}
\end{align*}
having a solution.  By a discussion in Section \ref{sec:computing_the_nullspace}, this is true if and only if $t~|~b_j$, $\forall j \in [m]$. Choosing $b$ as  such and keeping all entries non-negative, one solution is then $\hat{\mathbf{x}} = (V')^{-1}\left(\frac{b}{t}\right) \pmod{r\mathbb{Z}^n}$.  From Section \ref{ss:primal_group_problem}, ILP feasibility then hinges on
\begin{align*}
A^{-1}_{\Bcal}(b-A_{\Ncal}\hat{\mathbf{x}}) \geq 0,
\end{align*}
which holds by our assumptions on $A_{\Bcal}^{-1}$.

We also want to be able to plant our choice of basis $\Bcal$ as an optimal one. To perform the planting, we only need to ensure primal and dual feasibility of the corresponding LP basic feasible solution (BFS). The primal BFS is $x_{\textup{LP}} = (x_{\Bcal}, x_{\Ncal}) =(A_{\Bcal}^{-1}b, \mathbf{0}_{n-m})$ and the associated dual solution is $y_{\textup{LP}} = (A_\Bcal^{\top})^{-1}c_{\Bcal}$. By the earlier paragraph, we already have $x_{\Bcal} \geq 0$, ensuring primal feasibility. Dual feasibility requires the reduced costs to be nonnegative:  
$$
   c_\mathcal{N}^\top-c_\Bcal ^\top A_\Bcal ^{-1}A_\mathcal{N} \geq 0,
$$
which we can ensure by choosing $c$.

From equation \eqref{eq:K_isomorphism}, the null space $K$ has cardinality
\begin{align*}
    \lvert K \rvert = \prod_{j=1}^{m} \gcd (r, t) = t^m.
\end{align*}
Meanwhile, the cardinality of $G$ satisfies
\begin{align*}
    \lvert G \rvert = \prod_{i=1}^{m} \frac{r_m}{\gcd (r_m, t_i)} = \prod_{i=1}^m\frac{t^2}{t} = \abs{K}
\end{align*}
Hence for each $m$ we have a collection of ILPs corresponding to the set of all possible choices for $t, U, V, U', V'$, where $\lvert K \rvert = \lvert G \rvert$ for all of them.    
\end{proof}

To summarize, the only difference between the construction in Theorem \ref{thm:family_construction} and the set of all ILPs with $n = 2m$ is the requirement on the invariant factors of $A_{\Bcal}$ and $A_{\Ncal}$, i.e. the invariant factors of $A_{\Bcal}$ are all equal and each factor is exactly the square of a factor of $A_{\Ncal}$.

Note from \eqref{eqn:gdefin}, $G$ satisfies $\lvert G\rvert \leq \abs{\det A_{\Bcal}}$ and $\lvert \Acal \rvert \leq n$, which is why the shortest-path approach for the group relaxation is usually presented with an $\mathcal{O}(n |\det A_{\Bcal}|)$ runtime. Here we present a set of sufficient conditions on an ILP under which $\lvert K \rvert$ in the null space formulation is not significantly larger, i.e. in $\mathcal{O}(\textup{poly}(n, m)|\det A_{\Bcal}|)$.  Unlike the previous theorem, we relax some of the constraints on $n$ and the invariant factors of $A_{\Bcal}$. There are no conditions on $A_{\Ncal}$.
\begin{theorem}
    Let $c \in \mathbb{R}^n, A \in \mathbb{Z}^{m \times n}, b \in \Z{m}$ define an integer linear program, and $\Bcal$ be an optimal LP basis. Let $R = \diag(r_j)_j^m = UA_{\Bcal}V \in \mathbb{Z}^{m \times m}$ be the SNF of $A_{\Bcal}$. Suppose $n \leq  2m + \mathcal{O}(\log(m))$ and that there exists an integer $j =\mathcal{O}(\log(m))$ such that for all integers $k$ satisfying $j \leq k \leq m$, we have $r_j = r_m$. Then, we can construct a null space formulation of the ILP group relaxation where $\lvert K \rvert = \mathcal{O}(\textup{poly}(m, n) \abs{\det A_{\Bcal}})$. 
\end{theorem}
\begin{proof}
    We follow the procedure outlined in Section \ref{subsec:nullspaceform} with the domain $\mathbb{Z}_{r_m}^{d}$. Without loss of generality, we take $d = n-m$ here. By construction, $n -m \leq m + \mathcal{O}(\log(m))$. Note trivially $\lvert K \rvert \leq \lvert \mathbb{Z}_{r_{m}}^d \rvert$. Also $\lvert\det A_{\Bcal}\rvert = \prod_{i=1}^{m} r_i = \mathcal{O}(\textup{poly}(m)r_m^{m-j})$, which one can easily see is the same order of $\lvert \mathbb{Z}_{r_m}^{d} \rvert = r_m^d$.
\end{proof}
\noindent The main point is that $\lvert K\rvert$ tends to not be much larger than $|\det A_{\Bcal}|$ for problems with around twice as many variables as constraints and with an $A_{\Bcal}$ that has the majority of its invariant factors being equal.

\section{Quantum short path algorithm for the group relaxation}\label{s:quantum}
In this section we apply the quantum short path algorithm to the group relaxation. We provide two distinct instantiations of the framework. The first relies on the product-of-cycles walk (defined in Section \ref{subsec:local_search_nullspace}) and uses Theorem~\ref{thm:gen_short_path_runtime_ls}, while the second constructs an expander from the product-of-cycles walk and applies Theorem~\ref{thm:gen_short_path_runtime_poincare}. Before devling into these algorithms, we show that the quantum input assumption (Assumption~\ref{asm:input-assumptions}) is satisfied for both walks. 

\subsection{Quantum state preparation and block encoding}
\label{subsec:state_prep_and_block_encoding}
We recall below the input requirements of the SP algorithm.

\paragraph{Initial State Preparation} A critical prerequisite for the SP framework is the efficient preparation of $|\sqrt{\pi}\rangle$.  For loading the initial state for either the product of cycles walk or the expander walk, we can utilize the following procedure outline in Algorithm~\ref{alg:initial_state_prep}.
\begin{algorithm}  
   \caption{\textsc{State Preparation For Uniform Distribution over $K$}}   \label{alg:initial_state_prep}
    \textbf{Input:} Generating set $\Scal$\\
    \textbf{Output:} A quantum state $\ket{\sqrt{\pi}}$ encoding a uniform distribution over $K$
    \begin{algorithmic}[1]
    \State Construct a matrix $\Bmbf  \in \Zmbb ^{d \times k}$ using elements of $\Scal$ as  columns.
    \State Determine $u_j = \lvert \langle h_j \rangle\rvert$ for $h_j \in \Scal$
    \State Load superposition over $\bigoplus_i^k \Zmbb _{u_i}$
    \State Apply a unitary $\mathcal{U}_{\Bmbf }$ that performs $|\alpha\rangle|0\rangle \mapsto |\alpha\rangle|\Bmbf \alpha\rangle$
    \State Apply unitary  $\mathcal{U}_{\Bmbf _{\text{inv}}}$ that performs $|\alpha\rangle|\gamma\rangle |0\rangle \mapsto |\alpha\rangle|\gamma\rangle |\Bmbf ^{-1}\gamma\rangle$ 
    \State \textsf{CNOT} third (control) and first (target) registers, apply $\mathcal{U}_{\Bmbf _{\text{inv}}}^{\dagger}$ 
    \State Trace out first and third registers, i.e. the first and third registers will be definitely all zeroes (unentangled), hence we can ignore them.
    \end{algorithmic}
\end{algorithm} 
Step 3 can be carried out using $\mathcal{O}\left(k\max_{u_j \in \{ u_i : h_i \in \Scal\}}u_j\right)$ %\ba{What is $\Bcal_K$? Seems to conflate with the basis notation}\dyh{it's a typo, was some old notation.}
gates via \cite{shukla2023efficient,grover2002creating}. Steps 4 and 5 can be performed by reversibly encoding $\mathcal{O}(\text{poly}(n))$-sized classical circuits for applying $\Bmbf $ and $\Bmbf ^{-1}$. The rest of the operations are either classical polynomial time or require $\mathcal{O}(\poly(n))$ quantum gates.

\paragraph{Block-encodings of Markov Chain and Cost Function}
Furthermore, we require efficient block-encodings of both the Markov chain discriminant matrix $D(P)$ and the cost Hamiltonian $H$. We say that unitary $\mathcal{U}(M)$ is  $(\kappa, a)$-block encoding  of an $n$-qubit operator $M$ of
\begin{align*}
    (\langle 0|^{\otimes a} \otimes I)\mathcal{U}(M)(|0\rangle^{\otimes a} \otimes I) = M/\kappa.
\end{align*}
Specifically, for SP in general, we want to block-encode $H$ defined as
$$H \ket{\xmbf} =  f(\xmbf)\ket{\xmbf},$$ for the problem cost function $f$ \eqref{eq:transformed_cost_no_basis},
and  $D(P)$ defined by
$$D(P)_{\xmbf, \xmbf^{\prime}} = \sqrt{P(\xmbf,\xmbf^{\prime}) P(\xmbf^{\prime},\xmbf)},$$
when the chain $P$ is reversible.
% we assume the existence of quantum circuits implementing $(\kappa, a)$-block-encodings of these operators with $\kappa = \mathrm{poly}(n, m, \log p)$ and $a = O(1)$, as detailed in Appendix~D of~\cite{chakrabarti2024generalized}.\

%\paragraph{Cycle Product Walk}

Since Cayley walks form regular graphs, $P$ is symmetric and $D(P) = P$. Hence, we only need to block-encode $P$ (more specifically, its lazy variant). Note that there is an efficient classical circuit that computes $P(\xmbf,\xmbf^{\prime})$ for any $\xmbf, \xmbf^{\prime} \in K$, and hence we can efficiently construct  unitary operations that perform:
\begin{align*}
|\xmbf\rangle|j\rangle|0\rangle &\mapsto |\xmbf\rangle|j\rangle|r_{\xmbf j}\rangle\\
|\xmbf\rangle|j\rangle|0\rangle &\mapsto |\xmbf\rangle|j\rangle|c_{\xmbf j}\rangle\\
|\xmbf\rangle|\xmbf^{\prime}\rangle|0\rangle &\mapsto |\xmbf\rangle|\xmbf^{\prime}\rangle|P(\xmbf,\xmbf^{\prime})\rangle,
\end{align*}
where $r_{\xmbf j} \in K$ is such that $P(\xmbf, r_{\xmbf j})$ is the $j$-th nonzero entry of the row $P(\xmbf, \cdot)$, and 
$c_{\xmbf j} \in K$ is  such that $P(c_{\xmbf j}, \xmbf)$ is the $j$-th nonzero entry of the column $P(\cdot, \xmbf)$. Additionally, $P$ is $\lvert \Scal\rvert = \Theta(\log\lvert K \rvert) = \mathcal{O}(d\log r_m)$ row and column-sparse matrix. Hence, by \cite[Lemma 48]{gilyen2019quantum} we can prepare a $(\Theta(d\log r_m, \mathcal{O}(\text{poly}(n)))$-block encoding to $P$.
 
The cost function in this case is $f$ from Equation \eqref{eq:transformed_cost_no_basis}, which we can evaluate reversibly in quantum polynomial time. Specifically, the following is a $(\kappa, 1)$-block-encoding for $H$:
\begin{align*}
&\mathcal{U}(H): |\xmbf \rangle|0\rangle \mapsto \frac{f(\xmbf )}{\kappa} |\xmbf\rangle|0\rangle  + \sqrt{1- \frac{f(\xmbf )^2}{\kappa^2}}|\xmbf \rangle|1\rangle,
\end{align*}
where $\kappa$ ensures $\lvert f(\xmbf )/\kappa\rvert \leq 1$ and $\xmbf \in K$ (for example, $\kappa = \lVert \mathbf{c}_{1}\rVert \cdot r_m$ suffices.) This operation involves reversibly implementing $f$ and $\mathcal{O}(\log \kappa )$ additional gates to rotate $f(\xmbf )$ onto an amplitude. 

We already have the ingredients for a quadratic quantum speedup. As discussed in Theorem~\ref{thm:mcs_complexity}, classical Markov Chain search, with a fast mixing chain, requires time
$$\mathcal{O}((|K|/|K^\star|)\cdot \poly (n, L))$$
to find an optimal solution. Using Algorithm~\ref{alg:initial_state_prep}, we can efficiently prepare a quantum state encoding $\sqrt{\pi}$ for the Cayley walk. As a consequence, we trivially obtain a quadratic speedup for over MCS for this problem, via generalized quantum minimum finding \cite[Theorem 49]{van2020quantum}. We highlight this fact in a theorem:

\begin{theorem}[Quadratic Speedup for group relaxations] There is a quantum algorithm that, with high probability, finds an optimal solution to the group relaxation with runtime at most 
$$
    \mathcal{O}\left(\sqrt{|K|/ \lvert K^{\star}\rvert} \cdot \poly (n, L) \right).
    $$
\end{theorem}

Utilizing the generalized SP algorithm discussed in Section \ref{sss:runtime_short_path}, we can construct a quantum algorithm that improves upon the one above in the following sense. Our quantum algorithm achieves a runtime of $ \mathcal{O}((|K|/ \lvert K^{\star}\rvert)^{\frac{1}{2} - \alpha} \cdot \poly (n, L))$
for some $\alpha > 0$ determined by spectral properties of the walk over $K$, stability of the cost function (Definition \ref{def:delta-stability}), and spectral density (Definition \ref{defn:gamma-spectral}) parameters. While $\alpha$ may not be a dimension-independent constant in general, we do provide sufficient conditions for ensuring it is.

\subsection{Product of cycles walk approach}\label{ss:prod_cycles}
In this section we analyze the cost of the short path algorithm for solving the group relaxation when the underlying walk is the product-of-cycles walk described in Section~\ref{subsec:local_search_nullspace}.
% Our analysis requires the following intermediate result, which bounds the pseudo Lipschitz norm of our cost function under the product-of-cycles walk in terms of cyclic metrics. This bound serves as a proxy for $\Delta_P$-stability (see, Definition~\ref{def:delta-stability}).
% \begin{lemma}[Objective change under cyclic metrics]\label{lem:tight-step-weighted}
% Let $\xmbf\mapsto \xmbf^{\prime} = \xmbf + a h_j \bmod r_m \Z{d}$ be a Cayley step with $a\in\{\pm 1\}$ and $h_j\in\Scal$. Then
% $$
% \big|f(\xmbf^{\prime}) - f(\xmbf)\big|
% \;\le\;
% \cycnormonew{h_j}{\mathbf{c}},
% $$
% and, hence,
% $$
% \big(f(\xmbf^{\prime}) - f(\xmbf)\big)^2
% \;\le\;
% \cycnormonew{h_j}{\mathbf{c}}^{\;2}.
% $$
% \end{lemma}
% \begin{proof}
% The balanced difference obeys $|\xmbf_i'-\xmbf_i| \le \min\{|a h_{ji}|,\,r_m-|a h_{ji}|\}$ for each $i \in [d]$, where $h_{ji}$ is the $i$th coordinate of $h_j$ in balanced residues. The first inequality follows from
% $$
%     |f(\xmbf^{\prime})-f(\xmbf)| = \left|\sum_{i \in [d]} \mathbf{c}_i\,(\xmbf^{\prime}_i-\xmbf_i)\right|
%     \le \sum_{i \in [d]} |\mathbf{c}_i|\,|\xmbf^{\prime}_i-\xmbf_i|
%     \le \sum_{i \in [d]} |\mathbf{c}_i|\,\min\{|h_{ji}|,\,r_m-|h_{ji}|\}.
% $$
% This completes the proof up to trivial computation. Since $|a| = 1$, the “$\le$” to the unscaled $h_j$ holds in both cases. 
% \end{proof}
% We are now in a position to state the running time of the short path algorithm applied to group relaxations using a product-of-cycles walk. 
The runtime can be bounded as follows.
\begin{theorem}[Runtime for solving the group relaxation]
\label{thm:main-speedup}
With high probability, the generalized quantum SP algorithm finds the optimal solution to the group relaxation \eqref{e:GILO-final} with runtime at most
$$
\mathcal{O}\!\left(\left(\frac{|K|}{|K^\star|}\right)^{\frac{1}{2}-\alpha}\cdot \mathrm{poly}(n,L)\right).
$$
If the set $\Scal = \{h_1, \dots, h_{k}\}$ outputted by \textsc{NullGenFinding} (Algorithm \ref{alg:nullspacereform})  satisfies:
\begin{subequations}
    \begin{align}
\label{eqn:condition1}
    \max_{j\in[k]}\,\cycnormonew{h_j}{\mathbf{c}}
    &=
    \Theta\!\left(
    \frac{|{\OPT}_\Bcal|}
    {\log\left(|K|/\lvert K^{\star}\rvert\right)}
    \right),\\
    \label{eqn:condition2}
    \max_{j \in [k]}\lvert \langle h_j \rangle\rvert^2 &= \Theta\left(\frac{\log\left(|K|/\lvert K^{\star}\rvert\right)}{k}\right),
\end{align}
\end{subequations}
then $\alpha$ is a nonzero constant.
\end{theorem}
\begin{proof}
The quantum input assumption (see Assumption~\ref{asm:input-assumptions}) holds by results from Section~\ref{subsec:state_prep_and_block_encoding}. Recall from Eq. \eqref{eq:transformed_cost_no_basis} that $f: K \rightarrow \Rmbb_{\geq 0}$ is given by
\begin{align*}
    f(\xmbf) = \cmbf^\top (\hat{\xmbf} + \xmbf \pmod{r_m \Z{d}}) + c_\Bcal^{\top} A_\Bcal^{-1} b.
\end{align*}
Define a shifted cost $\widetilde{f}(\xmbf) \coloneqq f(\xmbf)-C$ for some $C=\Theta(\lvert \OPT_{\Bcal} \rvert)$ so that $E^\star=\min_{\xmbf \in K} \widetilde{f}(\xmbf)<0$.

By the definition of the $P$-pseudo-Lipschitz norm and $P$ being a Cayley walk:
\begin{align*}
\lVert \widetilde{f} \rVert_{P} &= \lVert f \rVert_{P} \\&= \frac{1}{3\lvert {\Scal} \rvert}\sum_{h_j \in {\Scal}}\sum_{a \in \{\pm 1, 0\}} (\langle \mathbf{c}, [\hat{\mathbf{x}} + \mathbf{x} + a\cdot h_j \pmod{r_m\mathbb{Z}^d}] - [\hat{\mathbf{x}} + \mathbf{x} \pmod{r_m\mathbb{Z}^d}]\rangle)^2\\
&=\frac{1}{3\lvert {\Scal} \rvert}\sum_{h_j \in {\Scal}}\sum_{a \in \{\pm 1, 0\}} \left(\sum_{i=1}^d c_i\left( [\hat{x}_{i} + x_i + a\cdot h_{ji} \pmod{r_m}] - [\hat{x}_i + {x}_i \pmod{r_m}]\right)\right)^2\\
&\leq \frac{1}{3\lvert {\Scal} \rvert}\sum_{h_j \in {\Scal}}\sum_{a \in \{\pm 1, 0\}} \left(\sum_{i=1}^d \lvert c_i\rvert \max\{ r_m - h_{ji}, h_{ji}\}\right)^2\\
&\leq \max_{h_j \in \Scal}  \cycnormonew{h_j}{\mathbf{c}} ^2.
\end{align*}
Also, as mentioned in Section~\ref{sss:runtime_short_path}, $\Delta_P \le \sqrt{\lVert \widetilde{f} \rVert_{P}}$.

From Theorem \ref{thm:gen_short_path_runtime_ls} and our assumptions, we take $\gamma$ for the spectral density condition to be
\begin{align}
\label{eqn:expression_for_gamma_proof}
    \gamma = \frac{\omega(-(1-\eta)[\OPT_{\Bcal} - C] - \mathbb{E}_{\pi}[\widetilde{\theta}_{\eta}(\widetilde{f})])^2}{(\max_{j\in[k]}\,\cycnormonew{h_j}{\mathbf{c}})^2\ln(\lvert K\rvert/\lvert K^{\star}\rvert))} = \Theta\left(\omega \log(\lvert K\rvert/\lvert K^{\star}\rvert)\right).
\end{align}
Note that in Theorem \ref{thm:gen_short_path_runtime_ls}, the expression for $\gamma$ contains $\mathbb{E}_{\pi}[H]$ in the numerator, where here $H = \widetilde{f}$. However, it is implicit in the analysis of \cite{chakrabarti2024generalized} that we can replace $\mathbb{E}_{\pi}[\widetilde{f}]$ with $\mathbb{E}_{\pi}[\widetilde{\theta}_{\eta}(\widetilde{f})]$. Note that by construction
\begin{align*}
    \OPT_{\Bcal} - C \leq \mathbb{E}_{\pi}[\widetilde{\theta}_{\eta}(\widetilde{f})] \leq (1-\eta) (\OPT_{\Bcal} - C).
\end{align*}
Since $\pi(E^{\star}) \neq 0$, we have that $\lvert -(1-\eta)(\OPT_{\Bcal} - C) - \mathbb{E}_{\pi}[\widetilde{\theta}_{\eta}(\widetilde{f})]\rvert$ is nonzero and $\in \Theta(\lvert  \OPT_{\Bcal} \rvert)$. Hence the expression in \eqref{eqn:expression_for_gamma_proof} follows.

We will want $\gamma = \Theta(1)$. However, note that we will also want
\begin{align*}
    \mu^{\star} = \gamma\omega\ln\left(\lvert K \rvert/ \lvert K^{\star}\rvert\right) = \Theta(\omega^2 \log^2\left(\lvert K \rvert/ \lvert K^{\star}\rvert\right))
\end{align*}
from Theorem \ref{thm:gen_short_path_runtime_ls} to be $\Theta(1)$. Hence both conditions require $\omega^{-1} = \Theta\left(\log\left(\lvert K \rvert/ \lvert K^{\star}\rvert\right)\right)$. By the discussion in Section \ref{subsec:local_search_nullspace}, \begin{align*}
    \omega =\Omega \left(\frac{1}{k \min_{j\in[k]} u_j^2 }\right)
\end{align*}
for the product-of-cycles walk, where $u_i \coloneqq \lvert \langle h_i \rangle\rvert$.
Hence our assumption implies both $\gamma$ and $\mu^{\star}$ are $\Theta(1)$.

We also need to check that condition \eqref{eqn:ell_condition_ls} is satisfied. By the above, it suffices for
\begin{align*}
    &\frac{1}{3\gamma \ln(1/\pi(E^{\star}))\sqrt{ \pi( E \leq (1-\eta)E^*)}} - \max\left(\frac{4\ln(\lvert K\rvert)}{\omega}, \frac{3|\OPT_{\Bcal} - C|^2(1-\eta)^2}{\Big(\max_{j\in[k]}\,\cycnormonew{h_j}{\mathbf{c}}\Big)^{\!2}}\right)\\
    &=\frac{\xi_1}{\ln(\vert K \rvert/\lvert K^{\star}\rvert) (\lvert K ^{\star}\rvert/\lvert K \rvert)^{\xi}} - \max\left(\xi_3, \xi_2\log^2(\lvert K \rvert/\lvert K^{\star}\rvert)\right)  \geq 2,
\end{align*}
for some constants $\xi, \xi_1, \xi_2, \xi_3 > 0$. This holds for sufficiently large problem size by our assumptions.
% \jw{This part of the argument feels slippery and hard to follow along. Suppose the constant in $\Theta(1)$ is 1000. What choices for $C, C'$ suffice? And how does this connect with \eqref{eqn:ell_condition_ls}? Separately, we should not use $C$ twice.}\dyh{Sorry there was a typo in the denominator. Now it should be clear why the constants don't matter? It's just replacing the variables in eqn 6 with their actual values.}\dyh{Ok added a few more details, since I realized I did this for the proof in  the appendix. lmk if this is fine.}\jw{I will check a little later then delete these comments, thanks.}

The runtime of the SP algorithm is
\begin{align*}
    \mathcal{O} \left(\textup{poly}(n)\omega^{-1}[\pi(E^{\star})^{-1}]^{\left(\frac{1}{2}-\frac{\eta(1-\eta)\lvert E^{\star}\rvert \mu}{2\ln(1/\pi(E^{\star}))\Delta_{P}}\right)}\right)
\end{align*}
with $\mu < \mu^{\star}$ then follows from Theorem \ref{thm:gen_short_path_runtime_ls}.

Fixing $\eta \in (0, 1)$ and using that $\mu^{\star}=\Theta(1)$, the super-quadratic condition becomes
$$
\frac{|E^\star|}{\Delta_P} \geq \frac{|E^\star|}{\max_{j\in[k]}\,\cycnormonew{h_j}{\mathbf{c}}}
\;=\;
\Theta\!\left(\frac{|\OPT_{\Bcal}|}{\max_{j\in[k]}\,\cycnormonew{h_j}{\mathbf{c}}}\right)
\;=\;
\Theta\!\left(\log (|K|/|K^\star|)\right),
$$ 
which holds by our assumptions.
Note that $\lvert \OPT_\Bcal - C\rvert =\Theta(|\OPT_{\Bcal}|)$.
Thus we have that $\alpha>0$ (cf.\ Corollary~\ref{cor:super_quad_runtime_short_path}) is  dimension-independent constant.  
This completes the proof.
% Similar to Theorem \ref{thm:main-speedup}, our assumptions and Corollary \ref{cor:super_quad_runtime_short_path} also imply that $\alpha > 0$ is a dimension interdependent constant.

\end{proof}

Here, we seek to make the super‑quadratic condition ($\alpha > 0$ constant) in Theorem~\ref{thm:main-speedup} more transparent. For the condition in \eqref{eqn:condition1}, we have the following bound on generator-weight norm:
  $$
     \max_{j\in[k]}\,\cycnormonew{h_j}{\mathbf{c}} \leq r_m \lVert \mathbf{c}\rVert_{\infty} \max_{j\in[k]}\lVert h_j\rVert_{\ell_0}.
  $$
Suppose that the generators are sparse, in the sense that $\max_{j\in[k]}\lVert h_j \rVert_{\ell_0} = \mathcal{O}(1)$.  This reduces the first condition to
\begin{align*}
    \frac{\OPT_\Bcal}{\lVert \mathbf{c} \rVert_{\infty}} =\Theta\!\big(r_m\log\left(\lvert K \rvert/\lvert K^{\star}\rvert\right)\big).
\end{align*}

Suppose   $\lvert Z \rvert = r_m^d = \Theta(\exp(n))$. This condition boils down to assuming the problem is hard, taking exponential time, but not super exponential time.  For hardness, we are assuming $m = \Theta(n)$. Thus these assumptions  yield $\max_{j \in [k]}\lvert \langle h_j \rangle\rvert^2  = \mathcal{O}(1)$. To see this, note by definition of $Z$, $\lvert Z \rvert = \Theta(\exp(n))$, and $n = \Theta(m)$ imply that $r_m = \lVert R \rVert_{\infty} = \mathcal{O}(1).$ However, $\lvert \langle h_j \rangle \rvert \leq r_m$.

Suppose also that reducing the problem to a search over $\lvert K \rvert$ does not result in subexponential runtime, i.e. $k=\lvert \Scal \rvert = \Theta(n)$. Thus, $\lvert K \rvert = \prod_{h_j \in \Scal} \lvert \langle h_j \rangle \rvert  = \Theta(\exp(n))$.  Thus, the second condition \eqref{eqn:condition2} is satisfied.

Thus we end up with the following more intuitive conditions for hard, yet not super-exponentially hard, problems. To summarize, $\alpha > 0$ is a nonzero constant if 
  \begin{enumerate}
      \item $n = \Theta(m)$;
      \item  the domain of the finite group relaxation $Z$, satisfies $\lvert Z\rvert = \Theta(\exp(n))$;
      \item $\frac{\OPT_\Bcal}{\lVert \mathbf{c}\rVert_{\infty}} = \Theta\!\big(n) $;
      \item restricting to the feasible set $\lvert K \rvert$ does not make the runtime subexponential;
      \item and the generating set $\Scal$ obtained from Algorithm \ref{alg:nullspacereform} is sparse : $\forall h_j \in \Scal, 
      \lVert h_j \rVert_{\ell_0} = \mathcal{O}(1)$.
  \end{enumerate}
One will note that most of these conditions are light, except for the condition on the sparsity of the generators.

The following result shows that there does exist a family of ILPs whose group relaxations satisfy the conditions for super-quadratic speedups using the product-cycle walk. The following ILPs turn out to have a specific structure that enables them to be solved easily. However, the goal is simply to use this construction to show that our conditions are not vacuous.
%\dyh{So these problems turn out to be easy once the structure is discovered (which is easy to discover). I think it's fine though and achieves the goal stated.}
\begin{theorem}
There exists a family of integer linear programming problems whose group relaxations satisfy conditions \eqref{eqn:condition1} and \eqref{eqn:condition2} in Theorem \ref{thm:main-speedup}. %Specifically, there is quantum algorithm whose worst-case asymptotic runtime over this family is $\mathcal{O}\left(\vert G\rvert^{\frac{1}{2} - \alpha}\cdot \poly(n, L)\right)$, where $\alpha > 0$ is a constant. Recall Gomory's shortest-path approach runs in $\mathcal{O}(\lvert G \rvert \cdot \poly(n, L))$.%\dyh{Might be able to just say super-quadratically faster than gomory's shortest path. Or maybe this is fine.}
\end{theorem}
\begin{proof}
The family is constructed using the procedure outlined in the proof of Theorem \ref{thm:family_construction}, with some additional assumptions on the unimodular matrix $V'$ chosen in the construction therein. To recall for the reader, the construction considers arbitrary positive integers $t, m$, with $n=2m$ and $r = t^2$. Define $R = r I_m $ and $T = t I_m $. We pick unimodular matrices $U, V, U', V \in \mathbb{Z}^{m \times m}$. The ILP constraint matrix is then $A = [URV | U'TV'] = [A_{\Bcal } | A_{\Ncal}] \in \mathbb{Z}^{m \times n}$, with cost $c$ chosen so that $\Bcal$ is an optimal basis. We now additionally specify $V' = I_m$.

Recall that our constraint on the cost function vector $c$ is from the need for the reduced costs to be nonnegative:
\begin{align*}
   c_\mathcal{N}^\top-c_\Bcal ^\top A_\Bcal ^{-1}A_\mathcal{N} \geq 0.
\end{align*}
This is satisfied by the choice $c_{\Bcal} = \mathbf{0}_m$ and $c_{\Ncal} = \mathbf{1}_{m}$. 
The conditions for super-quadratic speedup are
\begin{align*}
\max_{j\in[k]}\,\cycnormonew{h_j}{\mathbf{c}}=
    \Theta\!\left(
    \frac{|{\OPT}_\Bcal|}
    {\log\left(|K|/\lvert K^{\star}\rvert\right)}
    \right), \qquad \max_{j \in [k]}\lvert \langle h_j \rangle\rvert^2 = \Theta\left(\frac{\log\left(|K|/\lvert K^{\star}\rvert\right)}{k}\right).
\end{align*}
The group relaxation constraint is 
\begin{align*}
    U^{-1}A_{\Ncal}\xmbf  = U^{-1}b  \pmod{r\Z{m}}.
\end{align*}

Then let us follow the steps of Algorithm \ref{alg:nullspacereform}. Note that by construction, the SNF of $U^{-1}A_{\Ncal}$ is $U^{-1}U'TV'$. We first solve $T\zmbf = 0 \pmod{r\Z{m}}$, which from \eqref{eq:null_solution_set} has solution set
\begin{align*}
    \left\{ \ell t : \ell \in [t] \right\}^{\times m}.
\end{align*}
We then  transform the generating set $\{t \mathbf{e}_j : j \in [m]\}$ via $(V')^{-1}$ and mod by $r\Z{m}$, which is trivial here as $V'$ is the identity. Thus, $h_j = t \embf_j$, for each $j\in[m]$, and there are $k =m$ generators of $K$. These generators each have order $t$, thus $\max_{j \in [k]}\lvert \langle h_j \rangle\rvert^2 = t^2 = \Theta_{m}(1)$. Additionally, a direct calculation shows that for any $t \geq 2$,
\begin{align*}
    \cycnormonew{h_j}{\cmbf} = c_j (t^2 - t) = \Theta_m(1).
\end{align*}
Hence, $\max_{j\in[k]}\,\cycnormonew{h_j}{\cmbf} = \Theta_m(1)$.
% \begin{align*}
%    \max_{j\in[k]}\,\cycnormonew{h_j}{\mathbf{c}} &\leq (\max_{j \in [k]} \lvert \langle h_j\rangle \rvert)^2 \lVert \mathbf{c}\rVert_{\infty}^2\max_{h_j \in {\Scal}}\lVert h_j \rVert_{\ell_0}^2\\
% %     &\leq t^2 \max_{h_j \in {\Scal}}\lVert h_j \rVert_{\ell_0}^2,
% \end{align*}
% and we also have that $\max_{h_j \in {\Scal}}\lVert h_j \rVert_{\ell_0} = 1$. 

Thus, the two conditions for super-quadratic runtime become
\begin{align*}
    \log(\lvert K \rvert/\lvert K^{\star}\rvert) = \Theta(m), \qquad |\OPT_\Bcal| = \Theta(m).
\end{align*}

From the proof of Theorem \ref{thm:family_construction}, we can take $\hat{\mathbf{x}} = (V')^{-1}\left(\frac{b}{t}\right) \pmod{r\mathbb{Z}^m} = \left(\frac{b}{t}\right) \pmod{r\mathbb{Z}^m}$. Take $b = \ell t\mathbf{1}_m$, $\ell < r$. Thus our feasible coset has elements of the form
\begin{align*}
% &V'^{-1}\left(\frac{b}{t} + t\sum_{j}n_je_j\right)\pmod{t^2\mathbb{Z}^n}\\
% &=V'^{-1}\left(\ell \mathbf{1}_m + t\sum_{j}n_je_j\right)\pmod{t^2\mathbb{Z}^n}\\
&\ell + \sum_{j} t n_j\mathbf{e}_j \pmod{t^2\mathbb{Z}^n}, \quad n_j \in \{0,1,\ldots, t-1\}.
\end{align*}
%$n_j \in \{0,1,\ldots, t\}.$

This construction ensures that there is only one optimal solution, i.e. the element of the coset with smallest entries. We can pick $\ell$ to ensure that all the entries of the minimizer are nonzero. Hence we have $\lvert K^{\star} \rvert = 1$ and $|\OPT_\Bcal| = \Theta(m)$ by the chosen cost function. Because $\log \abs{K} = m \log t = \Theta_m(m)$, we see that for each $t \geq 2$ there is a family of problems satisfies both conditions of Theorem \ref{thm:main-speedup}.
\end{proof}

\subsection{Expander walk approach}\label{ss:SP_poin}

We can convert the product-of-cycles walk into an expander graph, with high probability, by randomly sampling additional elements from $K$. Specifically, if we use the minimal generating set $\Scal$ found using Algorithm \ref{alg:nullspacereform} to sample $C\log \lvert K\rvert$ (for some constant $C >0 $) elements $\widetilde{\Scal}$ uniformly at random from $K$, then with high-probability both $\widetilde{\Scal}$ will generate $K$ and the Cayley graph spanned by $\widetilde{\Scal}$ is an expander (Algorithm \ref{alg:expandergen}). The latter part follows from the \textit{Alon--Roichman theorem}~\cite{alon1994random}, which asserts that the spectral gap is $\delta \geq \epsilon$ for a constant $\epsilon> 0$. Since the Cayley walk that we generate is with high probability a spectral expander, we refer to this approach as the \emph{Expander Walk}. 

\begin{algorithm}
  \caption{\textsc{ExpanderGeneration}}
  \label{alg:expandergen}
  \hspace*{\algorithmicindent} \textbf{Input}: Constraints from Group relaxation  $(\Ambf , \{r_j\}_{j=1}^{m})$\\   
 \hspace*{\algorithmicindent} \textbf{Output}: A Cayley graph generated by a set $\widetilde{\Scal} \subset K$ s.t. with high probability the outputted set both generates $K$ and forms a Cayley graph that is a spectral expander.
  \begin{algorithmic}[1]
    \State Let $\Scal =\textup{NullGenFinding}(\Ambf , \{r_j\}_{j=1}^{m})$, compute $\lvert K \rvert$.
    \State Use $\Scal$ to uniformly sample $C\log\lvert K \rvert$ elements from $K$ and obtain $\widetilde{\Scal} \subset K$.
    \State Output $\widetilde{\Scal}$.
  \end{algorithmic}
\end{algorithm}

Since the expander has a constant spectral gap, we obtain an alternative runtime bound for solving the group relaxation. Compared with Theorem~\ref{thm:main-speedup}, the condition~\eqref{eqn:condition2} for $\alpha$ being dimension-independent is no longer required. 

\begin{theorem}[Runtime for solving the group relaxation -- Expander Version]\label{thm:main-speedup-poincare}
With high probability, the generalized quantum SP algorithm finds the optimal solution to the group relaxation \eqref{e:GILO-final} with runtime at most
$$
\mathcal{O}\!\left(\left(|K|/|K^\star|\right)^{\frac{1}{2}-\alpha}\cdot \mathrm{poly}(n,L)\right).
$$
If the set $\Scal=\{h_1,\dots,h_k\}\subset Z$ output by \textsc{ExpanderGeneration} (Algorithm~\ref{alg:expandergen}) is a generating set for $K$ satisfying
$$
\max_{j\in[k]}\,\cycnormonew{h_j}{\mathbf{c}}
\;=\;
\Theta\!\left(
\frac{|\OPT_\Bcal|}
{\log\!\left(|K|/|K^\star|\right)}
\right),
$$
then $\alpha$ is a positive constant independent of the problem dimension.
\end{theorem} 
\begin{proof}
The input assumptions in Assumption~\ref{asm:input-assumptions} hold by construction (see, Section~\ref{subsec:state_prep_and_block_encoding}). The Cayley walk induced by $\widetilde{\Scal}$ returned by \textsc{ExpanderGeneration} has a constant spectral gap $\delta=\Omega(1)$ with high probability.  Hence, we have that $\mu^\star=\Theta(\delta)=\Omega(1)$. Let $\widetilde{f}$ be as in the proof of Theorem \ref{thm:main-speedup}.

By Theorem~\ref{thm:gen_short_path_runtime_poincare}, for any $\mu<\mu^\star$, the short path algorithm runs in time
\begin{equation}\label{e:sp_run_time_proof}
    \mathcal{O}\!\left(\mathrm{poly}(n)\,\delta^{-1}\,\big[\pi(E^\star)^{-1}\big]^{\left(\frac{1}{2}-\frac{\eta(1-\eta)\,|E^\star|\,\mu}{2\,\ln(1/\pi(E^\star))\,\sqrt{\|\widetilde{f}\|_{P}}}\right)}\right),
\end{equation}
if $\frac{\delta}{2} > \pi(E^{\star})^{\gamma}$ and  
\begin{align*}
\frac{1}{2\sqrt{\pi(E \leq (1-\eta)E^*)}} - \max\left(\frac{4\ln(\lvert K\rvert)}{\delta}, \frac{3\lvert E^{\star}\rvert^2(1-\eta)^2}{\Delta_P^2}\right)  \geq 2.
\end{align*}

Like in the proof of Theorem \ref{thm:main-speedup}:
$$
\Delta_{P} \leq \|\widetilde{f}\|_{P} \;=\; \|f\|_{P}
\;\le\;
\Big(\max_{j\in[k]}\,\cycnormonew{h_j}{\mathbf{c}}\Big)^{\!2},
$$
which holds because the expander walk is still a Cayley walk.
% Also, as mentioned in Section \ref{sss:runtime_short_path}, $\Delta_P \le \max_{j\in[k]}\,\cycnormonew{h_j}{\mathbf{c}}$.

Also similar to  the proof of Theorem \ref{thm:main-speedup}, but using the expression for $\gamma$ from Theorem \ref{thm:main-speedup-poincare}, we have
\begin{align*}
    \gamma = \frac{\sqrt{\delta}(-(1-\eta)(\OPT_\Bcal - C) - \mathbb{E}_{\pi}[\widetilde{\theta}_{\eta}(\widetilde{f})])}{\max_{j\in[k]}\,\cycnormonew{h_j}{\mathbf{c}}\ln(\lvert K\rvert/\lvert K^{\star}\rvert))} = \Theta(1),
\end{align*}
with $C = \Theta(\lvert \OPT_\Bcal\rvert)$ and by our assumptions. Hence,  since $\delta = \Omega(1)$, we have $\frac{\delta}{2} > \pi(E^{\star})^{\Theta(1)}$ for sufficiently large problem size, implying the first condition is satisfied. The second condition reduces to
\begin{align*}
&\frac{1}{2\sqrt{\pi(E^*)^{\Theta(1)}}} - \max\left(\frac{4\ln(\lvert K\rvert)}{\delta}, \frac{3|\OPT_\Bcal - C|^2(1-\eta)^2}{\Big(\max_{j\in[k]}\,\cycnormonew{h_j}{\mathbf{c}}\Big)^{\!2}}\right)\\
&\quad= \frac{1}{2\sqrt{(\lvert K ^{\star}\rvert/\lvert K \rvert)^{\xi}}} - \max\left(\xi_2\ln(\lvert K\rvert), \xi_3\log^2(\lvert K\rvert/\lvert K^{\star} \rvert)\right) \geq 2,
\end{align*}
by assumption and for some constants $\xi, \xi_2, \xi_3> 0$.  The second condition is satisfied for large enough problem size. Thus we have the runtime stated in Theorem \ref{thm:gen_short_path_runtime_poincare}.

Using that $\mu < \mu^{\star}=\Theta(1)$, the super-quadratic condition becomes
$$
    \frac{|E^\star|}{\Delta_P} \geq \frac{|E^\star|}{\max_{j\in[k]}\,\cycnormonew{h_j}{\mathbf{c}}}
    \;=\;
    \Theta\!\left(\frac{|\OPT_\Bcal|}{\max_{j\in[k]}\,\cycnormonew{h_j}{\mathbf{c}}}\right)
    \;=\;
    \Theta\!\left(\log\frac{|K|}{|K^\star|}\right),
$$
which is satisfied by our assumptions.
 Thus we have that $\alpha>0$ (cf.\ Corollary~\ref{cor:super_quad_runtime_short_path}) is a dimension-independent constant. 
%  We also have
% \begin{align*}
%     \gamma = \frac{\sqrt{\delta}(-(1-\eta)\widetilde{\OPT}_\Bcal - \mathbb{E}_{\pi}[\widetilde{\theta}_{\eta}(\widetilde{f})])}{\sqrt{B}\ln(\lvert K\rvert/\lvert K^{\star}\rvert))} = \Theta(1),
% \end{align*}
% by the above. By construction, $\lvert \widetilde{\OPT}_\Bcal \rvert = \mathcal{O}( \lvert \mathbb{E}_{\pi}[\widetilde{\theta}_{\eta}(\widetilde{f})]\rvert)$. Thus, the condition for super-quadratic runtime by Corollary \ref{cor:super_quad_runtime_short_path} is met.
This completes the proof.
\end{proof}

\section{Numerical results}
%Adding a temporary section here on the numerics
\label{sec:numerical}

We evaluate our algorithm on two complementary problem classes: a synthetic suite of CUTGEN1 one‑dimensional cutting‑stock instances, which lets us generate many reproducible problems via parametric sweeps \cite{GAU1995572}; and
(ii) a curated set of pure‑integer ILPs from MIPLIB~2017, providing diverse, real‑world benchmarks \cite{miplib2017}.
This pairing allows us to test behavior both in controlled settings (generation at scale) and across accepted benchmark models.
Because the group relaxation is defined for any ILP and yields bounds stronger than the plain LP relaxation, with the minimization chain $\OPT_\mathrm{LP} \le \OPT_\Bcal \le \OPT$, these two families are natural testbeds.

\begin{figure}[h]
    \centering
    \pgfplotstableread[col sep=comma]{gau_benchmark_table_combined.csv}\datatable

\begin{tikzpicture}
\begin{axis}[
    width=12cm,
    height=7cm,
    ybar,
    bar width=9.5,
    xmin=0, xmax=100,
    ymin=0,
    xlabel={$R\% = 100 \cdot \frac{\OPT_{\Bcal} - \OPT_{\mathrm{LP}}}{\OPT - \OPT_{\text{LP}}}$},
    ylabel={Count},
    xtick={0,20,40,60,80,100},
    xticklabels={0\%,20\%,40\%,60\%,80\%},
    grid=both, 
]

\addplot+[
    hist={
        bins=10,
        data min=0,
        data max=100.01
    },
    fill=blue!60,
    draw=blue!80!black,
] table [y=R_pct] {\datatable};

\draw[dashed, black, thick] (95,0) -- (95,27);
\node[anchor=south, font=\small] at (95,27.5) {100\%};

\end{axis}
\end{tikzpicture}
    \caption{Histogram of $R_{\%} = 100 \cdot (\OPT_\Bcal - \OPT_{\mathrm{LP}})/(\OPT - \OPT_{\mathrm{LP}})$ over the cutting‑stock benchmark. Dashed line at $100\%$ marks full gap closure by the group relaxation.}
    \label{fig:cutstock_hist}
\end{figure}
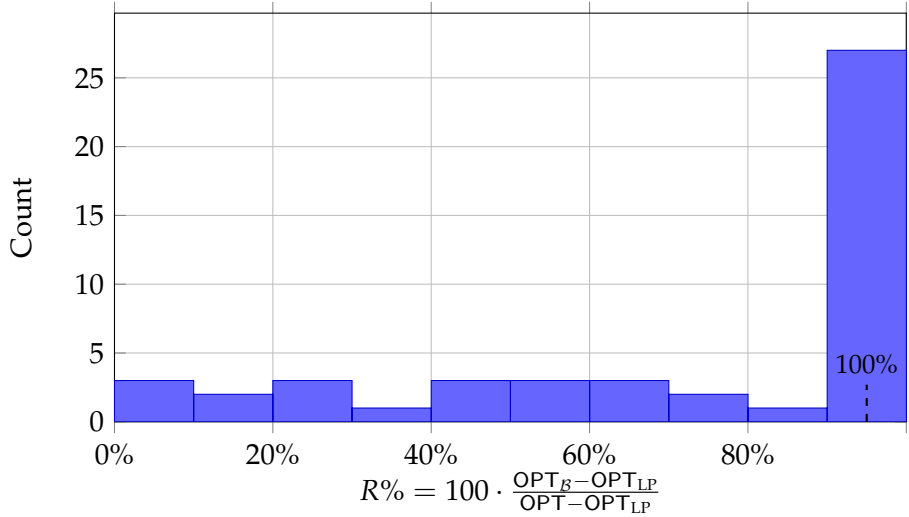

We use the standard pattern–based ILP for the 1D–CSP: given stock length $L$, item widths $w_i$ and demands $d_i$ for $i\in[m]$, let $P$ be the set of cutting patterns with counts $a_{ip}\in\mathbb{Z}_{\ge 0}$ satisfying $\sum_{i=1}^m a_{ip} w_i \le L$. The problem at hand is thus
$$
\min_{x\in\mathbb{Z}_{\ge 0}^{|P|}} \left\{ \;\sum_{p\in P} x_p
 : 
\sum_{p\in P} a_{ip}\, x_p \;\ge\; d_i \;\; \forall i\in[m] \right\},
$$
i.e., minimize the number of stock rolls subject to satisfying all demands.

Instance generation follows the CUTGEN1 scheme of Gau and W\"ascher~\cite{GAU1995572}. We instantiate a grid over three knobs
$$
m\in\{25,50,75\},\qquad
v_2\in\{0.25,0.50,0.75,1.00\},\qquad
\overline{d}\in\{5,10,20\}.
$$
Here: $m$ is the number of distinct item types;
 $v_2$ is the upper width ratio relative to the stock length $L$, and widths are sampled as integers in an interval proportional to $L$, typically
  $$
  w_i \in \big[\lceil v_1 L\rceil,\;\lfloor v_2 L\rfloor\big],
  $$
  where $v_1$ is a fixed lower ratio that is held constant across our runs.

\begin{table}[h]
\centering
% \caption{Group vs ILP: gaps and reduction }
\label{tab:grp_ilp_uniform}
\renewcommand{\arraystretch}{1.12}
\begin{tabular}{lrrrrrrr}
\toprule 
\hline
\textbf{Instance} &
$\mathsf{OPT_{\mathrm{LP}}}$ &
$\mathsf{OPT_\Bcal}$ &
$\mathsf{OPT_{\mathrm{ILP}}}$ &
$\boldsymbol{\Delta_{\mathrm{LP}\to \mathrm{ILP}}}$ &
$\boldsymbol{\Delta_\Bcal}$ &
$\boldsymbol{R_\mathrm{abs}}$ &
$\boldsymbol{R_\%}$ \\
\hline
\texttt{qap10}   & 332.57   & 336.00   & 340.00   &  7.43 &  4.00 &  3.43 & 46.2\% \\
\texttt{gen02}   & -4840.54 & -4817.32 & -4783.73 & 56.81 & 33.59 & 23.22 & 40.9\% \\
\texttt{ex10}    & 100.00   & 100.00   & 100.00   &  0.00 &  0.00 &  0.00 & N/A    \\
\texttt{supp19}  & 12677206.00 & 12677206.00 & 12677552.00 & 346.00 & 346.00 & 0.00 & 0.0\% \\
\bottomrule
\end{tabular}
\caption{Minimization results on pure ILP instances extracted from MIPLIB 2017~\cite{miplib2017}. Column definitions: $\Delta_{\mathrm{LP}\to\mathrm{ILP}}:=\OPT_\mathrm{ILP}-\OPT_\mathrm{LP}$, $\Delta_\Bcal:=\OPT_{\mathrm{ILP}}-\OPT_\Bcal$, $R_\mathrm{abs}:=\OPT_\Bcal-\OPT_\mathrm{LP}$, and $R_\%:=100\cdot R_\mathrm{abs}/\Delta_{\mathrm{LP}\to\mathrm{ILP}}$.}

\end{table}

\section*{Acknowledgements}
 The authors thank Rob Otter for executive support and valuable feedback on this project. We also acknowledge our colleagues at the Global Technology Applied Research Center
of JPMorganChase.

\bibliographystyle{alpha}
\bibliography{sdpbib}

\section*{Disclaimer}
This paper was prepared for informational purposes by the Global Technology Applied Research center of JPMorgan Chase \& Co. This paper is not a product of the Research Department of JPMorgan Chase \& Co. or its affiliates. Neither JPMorgan Chase \& Co. nor any of its affiliates makes any explicit or implied representation or warranty and none of them accept any liability in connection with this paper, including, without limitation, with respect to the completeness, accuracy, or reliability of the information contained herein and the potential legal, compliance, tax, or accounting effects thereof. This document is not intended as investment research or investment advice, or as a recommendation, offer, or solicitation for the purchase or sale of any security, financial instrument, financial product or service, or to be used in any way for evaluating the merits of participating in any transaction.
\newpage
\appendix

\section{Compressing the search space}
\label{app:kernel_compression}
This appendix discusses an approach to shrink the kernel group $K$ by compressing the ambient space. Recall from equation \eqref{eq:kernel_set_defn} that $K$ is a subgroup of $Z = \Zmbb_{r_m}^d$. Though this set is finite, it is clearly not the minimal choice. To see this, let $s_j = \abs{\langle\Ambf_j\rangle}$ be the order of $\Ambf_j$. Then, for any $\xmbf \in K$ the coordinate $\xmbf_j$ can always be reduced modulo $s_j \coloneqq \abs{\langle \Ambf_j \rangle}$ while remaining in constraint. Explicitly, the orders are given by
\begin{align*}
    s_j = \textup{lcm}\left(\frac{r_1}{\gcd(r_1, \Ambf_{1j})}, \dots, \frac{r_m}{\gcd(r_m, \Ambf_{mj})}\right).
\end{align*}
To see this, observe that any solution to $s_i \Ambf_{ij} = 0 \pmod{r_i}$ is a multiple of $\gcd(r_i, \Ambf_{ij})$. Any solution to these equations simultaneously must therefore be a common multiple. By definition of order, the smallest $s_j$ satisfying these equations is therefore the \emph{least} common multiple. 

Note that $r_m$ is a common multiple of the $\gcd(r_i, A_{ij})$, hence $s_j \vert r_m$. This implies a natural homomorphism 
\begin{equation*}
    Z \rightarrow Z' \coloneqq \bigoplus_{i=1}^d \Zmbb_{s_i}
\end{equation*}
obtained by reducing each component modulo $s_i$. The kernel $K$ can be reduced under this map to some $K' \leq Z'$. In this way, we obtain a new, generally smaller kernel space containing the original optimal solutions. This procedure provides a benefit when \begin{align*}
    \prod_{j=1}^{d}s_j \ll r_m^d.
\end{align*}

A set of generators for $K'$ can be obtained from the generators $(h_i)$ of $K$ by $h_i' = h_i \pmod{Z'}$. However these generators do not necessarily provide a direct sum decomposition of $K$. Such a decomposition is desirable for several reasons. At a high level, it allows for efficient search over the group elements without repetitions. More specifically to our purposes, we will find the cyclic decomposition useful for considering Markov chain approaches: the Cayley graph becomes a product of cycles, which has provably nice mixing properties.

To obtain cyclic generators for $K'$, we utilize Algorithm 4.1.10 of Cohen \cite{cohen2000advanced} for computing $K'$. The set up is a follows: let $\varphi(x) = x\pmod{Z'}$ be the homomorphism mapping $Z$ to $Z'$ (and thus $K$ to $K'$). We note that $\varphi$ is ``effective," meaning images and pre-image elements are trivial to compute. The domain space $Z$ is already in SNF (as a direct sum of cyclic groups, each order dividing the previous) , while the codomain $Z'$ can be put in SNF efficiently (especially given that the corresponding group relation matrix $\diag(s_j)_j$ is diagonal). Let $(\alpha_1,\ldots, \alpha_q)$ be the cyclic generators of $Z$ with invariant factors $z_1, \ldots z_q$, $q\leq d$. We note that $\varphi$ is still effective with respect to the new generators $(\alpha_j)_{j=1}^q$. 

The last ingredient is an effective description of $K$ as a subgroup of $Z$. In particular, by Prop. 4.1.6 of \cite{cohen2000advanced}, 
% \dyh{I couldn't find any algorithm in the ref, are there multiple versions?} \jw{Oh okay, I will investigate this.}
% Should have been 4.1.6, now fixed
any subgroup $A$ of $B$ is in 1-1 correspondence with integral matrices $H_A$ in \emph{Hermite Normal Form} (HNF). Such a matrix can be found in time polynomial in $d$, and typically faster than SNF. For our present case, $H_K$ is the matrix obtained from an HNF algorithm on the block matrix
\begin{equation}
    (h_1 h_2 \ldots h_\ell \vert r_m I_{d\times d}).
\end{equation}
We thus assume access to the approriate descriptor $H_K$ of $K$ in $Z$. 

With $H_K$ in hand, Algorithm 4.1.10 outputs another HNF integer matrix $L$ such that
\begin{equation}
    (\alpha_1 \ldots \alpha_q) L = (h_1' \ldots h_p')
\end{equation}
where $(h_i')_{i=1}^p$ form the cyclic generators of $K'$ and $p \leq \ell$. The runtime for this Algorithm is dominated by HNF calculations, which are again polynomial in $d$ and the size of the invariant factors.

Here we also present an alternative and self-contained approach (Algorithm \ref{alg:null_gen_finding_compress}) for compressing the kernel. This again produces a set of cyclic generators for $K'$.
\begin{algorithm}
  
\label{alg:null_gen_finding_compress}\caption{\textsc{NullGenFindingCompress}}
  \hspace*{\algorithmicindent} \textbf{Input}: Constraints from Group relaxation  $(\Ambf , \{r_j\}_{j=1}^{m})$\\   
 \hspace*{\algorithmicindent} \textbf{Output}:   An independent generating set $\Scal$ 
 for ${K}'\leq \bigoplus_{j=1}^d \mathbb{Z}_{s_j}$, the solution space to $\Ambf \xmbf = \mathbf{0} \pmod{R\Z{m}}$, where $s_j = \lvert \langle \mathbf{A}_{:j}\rangle \rvert$. 
  \begin{algorithmic}[1]
    \State Let $\{h_1, \dots, h_k\} =\textup{NullGenFinding}(\Ambf , \{r_j\}_{j=1}^{m})$
    \State Compute $s_j = \lvert \langle \mathbf{A}_{j,:} \rangle \rvert$, the orders of the columns of $\mathbf{A}$ as elements of $\mathbb{Z}_{r_m}^m$.
    \State Construct $B = \text{diag}(r_m/s_1, \dots, r_m/s_k)$, $D = [h_1, \dots, h_k]$.
    \State Let $\{\widetilde{m}_1, \dots, \widetilde{m}_{t}\} =\textup{NullGenFinding}(B{D} , \{r_m\}_{j=1}^{m})$.
    \State Let $m_t = D\widetilde{m}_t\pmod{r_m\mathbb{Z}^d}$.
    \State Compute SNF of $\mathbf{C} = [ m_1, \dots, m_t | r_m e_1, \dots r_m e_d] = UMV$.

    \State Let $\Scal'$ include columns of $U$  such that $M_{jj} \notin \{0, 1\}$.
    \State Let $\Scal = \{ u \pmod{S\mathbb{Z}^d}~|~ u \in \Scal'\}$.
    \State Output ${\Scal}$.
  \end{algorithmic}
\end{algorithm}

\begin{lemma}[Runtime of \textsc{NullGenFindingCompress}]
    Suppose one has data $(\Ambf , \bmbf , \mathbf{c}, \{r_j\}_{j=1}^{m})$ corresponding to the group relaxation of an ILP stored in classical memory. Let $S = \textup{diag}(s_1, \dots, s_d)$, where $s_j$ is the order of the $j$-th column of $\mathbf{A}$ taken  $\mod{R\mathbb{Z}^m}$. There is a classical algorithm that outputs cyclic generators $\Scal = (\widetilde{h}_i)_{i=1}^v$ for $\widetilde{K}$, the kernel of $\Ambf$ in $\bigoplus_{j=1}^d \mathbb{Z}_{s_j}$. In particular, ${K}' \cong \Zmbb _{\lvert \langle \widetilde{h}_1\rangle\rvert} \oplus \cdots \oplus \Zmbb _{\lvert \langle \widetilde{h}_v\rangle\rvert}$. The algorithm uses time $\widetilde{\mathcal{O}}(\textup{poly}(\log\lVert \Ambf  \rVert_{\infty}, d) )$. 
\end{lemma}
\begin{proof}
    We start by proving the validity of the above procedure. We again consider the generators output by Algorithm \ref{alg:nullspacereform}, $(h_j)_{j=1}^{k}$, which are all in $Z = \mathbb{Z}_{r_m}^d$. 

Recall $S$ is a diagonal matrix with the orders of the columns of $\mathbf{A}$ modulo $R\mathbb{Z}^m$. Note that reducing the element modulo $S\mathbb{Z}^d$ still produces a solution to $\Ambf \xmbf = \mathbf{0} \pmod{R\Z{m}}$.

Let $D = [h_1, \dots, h_k]$. We then want to solve
\begin{align*}
    D\mathbf{x} = \mathbf{0} \pmod{S\mathbb{Z}^d},
\end{align*}
which is to find the kernel of the map $\mathbf{y} \mapsto \mathbf{y} \pmod{S\mathbb{Z}^d}$ in $\langle h_1, \dots, h_k\rangle \leq Z$. To do this, we solve a scaled system. We know that for each integer $j, 1\leq j \leq d, s_j~|~r_m$. Let $B = \text{diag}(r_m/s_1, \dots, r_m/s_k)$. We then compute  a generating set for the solution set to
\begin{align*}
    BD\mathbf{x} = 0 \pmod{r_m\mathbb{Z}^d},
\end{align*}
i.e. run algorithm \ref{alg:nullspacereform} getting generators $\widetilde{m}_1, \dots, \widetilde{m}_t \in Z$. We then know that the set of $m_t = D\widetilde{m}_t \pmod{r_m\mathbb{Z}^d}$ span the kernel of the map $\mathbf{y} \mapsto \mathbf{y} \pmod{S\mathbb{Z}^d}$ in $\langle h_1, \dots, h_k\rangle \leq Z$

Let us now work in the coset space $ \mathbb{Z}^d/r_m\mathbb{Z}^d \cong Z$: we have
\begin{align*}
    \widetilde{Z} := \langle h_1 + r_m\mathbb{Z}^d\rangle \oplus \cdots \oplus \langle h_k + r_m\mathbb{Z}^d\rangle \leq \langle \mathbf{e}_1 + r_m \mathbb{Z}^d \rangle \oplus \cdots \langle \mathbf{e}_d + r_m \mathbb{Z}^d \rangle = \mathbb{Z}^d/r_m\mathbb{Z}^d.
\end{align*}

We also have that
\begin{align*}
    \widetilde{K} := \langle m_1 + r_m\mathbb{Z}^d \rangle \oplus \cdots \oplus \langle m_t + r_m\mathbb{Z}^d\rangle 
\end{align*}
forms the kernel of $ \mathbf{y} \mapsto \mathbf{y} \pmod{S\mathbb{Z}^d}$ in  $\langle h_1 + r_m\mathbb{Z}^d\rangle \oplus \cdots \oplus \langle h_k + r_m\mathbb{Z}^d\rangle$.

Following \cite[Section 6.2]{childs2017lecture}, let 

\begin{align*}
    \mathbf{C} = [ m_1, \dots, m_t | r_m e_1, \dots r_m e_d].
\end{align*} We then compute the SNF of $\mathbf{C}$ to get $\mathbf{C} = UMV$. By unimodularity of $U \in \mathbb{Z}^{d \times d}$ we have
\begin{align*}
     \widetilde{Z} \leq \langle u_1 + r_m \mathbb{Z}^d \rangle \oplus \cdots \oplus\langle u_d + r_m \mathbb{Z}^d \rangle = \mathbb{Z}^d/r_m\mathbb{Z}^d,
\end{align*}
where $u_j = U_{: j}$. By construction $M_{jj}u_j \in \widetilde{K}$. Since our goal is to find a generating set for the range of $ \mathbf{y} \mapsto \mathbf{y} \pmod{S\mathbb{Z}^d}$, we can ignore $u_j$ where $M_{jj} = 1$. 

We can also drop $u_j$ such that $M_{jj} = 0$. This is because for such $u_j$, $\langle u_j + r_m\mathbb{Z}^d \rangle$ is not in $\widetilde{Z}$. To see this: first, note that $M_{jj} = 0$ and the columns of $U$ are linearly independent, so we have that for any  integer $k, 0 <  k < r_m,  k u_j \notin \widetilde{K}$. The only option then is either $u_j \notin \widetilde{K}$ or $\langle u_j + r_m\mathbb{Z} \rangle \in \widetilde{Z}$, i.e.  the entire cyclic subgroup of $u_j$ is in $\widetilde{Z}$, no part is mapped to the kernel.

Suppose the later. Then we know that if we delete the all-zero columns from $M$ to get $M'$ then $r_mu_j \in \textup{colspace}_{\mathbb{Z}}(UM')$. However, the columns of $U$ are linearly independent and form a basis over $\mathbb{R}^d$, a contradiction. Hence $u_j \notin \widetilde{K}$.

This implies 
\begin{align*}
 \widetilde{Z}/\widetilde{K} = \langle (u_1 + r_m\mathbb{Z}^d) +  \widetilde{K} \rangle \oplus \cdots \oplus \langle (u_v + r_m\mathbb{Z}^d) +  \widetilde{K} \rangle .    
\end{align*}
One will also note that we then immediately have
\begin{align*}
    \widetilde{Z}/\widetilde{K} \cong \langle u_1 \pmod{S\mathbb{Z}^d}\rangle \oplus \cdots \oplus \langle u_v \pmod{S\mathbb{Z}^d} \rangle = K'  \leq \bigoplus_{j=1}^{d} \mathbb{Z}_{s_j},
\end{align*}
as desired.

% We can then take the collection of $\Scal = \{\widetilde{h}_j = u_j \pmod{S\mathbb{Z}^d} \}$ to be our generating set. Here, $\Scal$ generates the kernel $\widetilde{K} \cong \mathbb{Z}_{\lvert \langle \widetilde{h} _j\rangle\rvert} \oplus \cdots \oplus  \mathbb{Z}_{\lvert \langle \widetilde{h}_v \rangle\rvert}$ of $\mathbf{A}$ in $\bigoplus \mathbb{Z}_{s_j}$.

The runtime follows from the runtime presented for SNF in Section \ref{subsec:nullspaceform}.

% This gives a generating set forming a direct sum for a subgroup of $\mathbb{Z}_{r_m}^d$ such that no generator has order larger than $\textup{lcm}(s_1, \dots, s_d) \leq r_m$.
% Additionally, the cardinality of the group  (product of the orders of generators) is bounded by $\prod_j s_j$, each element has infinity norm bounded by $r_m$.

\end{proof}

\section{A Gibbs sampling approach for quantum speedup}\label{sec:alt_speedup}
In this section, we present a Gibbs sampler that comes from applying a Metropolis filter (i.e. Metropolis-Hastings algorithm) to either of the Cayley walks described in Sections~\ref{ss:prod_cycles} and \ref{ss:SP_poin}.

While Gibbs sampling for inverse-temperature $\beta > 0$ will outperform the uniform sampling approaches, the complexity of classical Gibbs sampling from the null space $K$ for $\beta > 0$ is unclear. The reason we still mention the Gibbs sampler, is that if one shows that classical Gibbs sampling at some $\beta > 0$ can be done classically in $\mathcal{O}(\textup{poly}(n))$ time and the conditions we present are met, then there is a super-quadratic speedup over the Gibbs sampler.

Beyond uniform sampling, one can also introduce a Metropolis filter (run the Metropolis-Hastings algorithm \cite{levin2017markov})  to perform Gibbs sampling over the feasible coset $\hat{\mathbf{x}} + K$ using a walk with the corresponding Gibbs measure as its stationary distribution. The Gibbs distribution over $\hat{\mathbf{x}} + K$ at inverse-temperature $\beta$ is
\begin{align*}
    \pi_{\beta}(x) \propto \exp\left(\beta \langle \mathbf{c}, x\rangle\right), x \in \hat{\mathbf{x}} + K,
\end{align*}
for $\beta > 0$. The Metropolis filter performs an accept-reject step on top of  another walk, referred to as the \emph{base chain}. Specifically, we can use either the expander or product cycle walk. One step of the Metropolis-adjusted chain is as follows:  starting from $x$, accept a move to $y \neq x$ proposed by the base chain with probability $\min\left(1, \exp\left(\beta\langle \mathbf{c}, y -x \rangle\right)\right)$.

For $\beta > 0$, the Gibbs sampler will provide a strict improvement over uniform sampling. Unfortunately, it appears there is not much known about the complexity of sampling for $\beta > 0$, implying that the current best theoretical result only applies to $\beta = 0$. %We refer to this walk as the \emph{Metropolis-adjusted Cayley Walk}.

% As mentioned in Section \ref{subsec:markov_chain_feas_set}, it is possible that Gibbs sampling at inverse-temperature $\beta > 0$ could provide a classical improvement over Theorem \ref{thm:mcs_complexity}. However, it appears that it is not well-understood what is the sampling complexity for $\beta > 0$. Still,

The below theorem shows that if there is a $\beta > 0$ such that the classical Metropolis chain satisfies a LSI with sufficiently large constant, then SP can provide a super-quadratic speedup over the Gibbs sampler. %From the proof, it will also be apparent that it suffices for the spectral gap of the Metropolis chain to be a constant, i.e. we could of used Theorem \ref{thm:gen_short_path_runtime_poincare}.
The next result follows from Theorem~\ref{thm:gen_short_path_runtime_ls} and Corollary~\ref{cor:super_quad_runtime_short_path}.
\begin{theorem}[Conditional Super-Quadratic Speedup Over Gibbs Sampling]
Suppose there exists a $\beta > 0$ such that a Metropolis-adjusted Cayley Walk with stationary distribution $\pi_{\beta}$ has LS constant lower bounded by $\omega = \Omega\left([\log(1/\pi_{\beta}(\OPT_\Bcal)]^{-1})\right)$. Then, under condition \eqref{eqn:condition1} from Theorem \ref{thm:main-speedup}, there exists a generalized quantum SP algorithm that runs in time
\begin{equation}
    \mathcal{O}\left( [\pi_{\beta}(\OPT_\Bcal)]^{-(1/2-\alpha)} \poly(n,L) \right),
\end{equation}
for $\alpha > 0$.
\end{theorem}
Note that the proof is effectively the same as Theorem \ref{thm:main-speedup} , as adding the Metropolis filter can only decrease $\Delta_{P}$, so the choices made for $\Delta_{P}$ in the proof of Theorem \ref{thm:main-speedup} remain valid. In addition, the Metropolis-adjusted Cayley walk can be performed with a constant number of queries to the base chain and cost function operators using  the approach of \cite{Lemieux2020efficientquantum} combined with the Cayley-walk implementation previously mentioned and Section \ref{subsec:state_prep_and_block_encoding}. The above theorem implies that if there exists a $\beta>0$ that satisfies the specified conditions, the runtime of both classical and quantum algorithms proposed in the main text can be improved as the distribution is tilted towards the optimal solutions. Unfortunately, for most general Gibbs measures, the distribution starts forming clusters and there is no local Markov chain that can efficiently sample from such a measure due to long range correlations.

\newpage
\section{Notations}

\renewcommand{\arraystretch}{1.25}

\begin{longtable}{p{0.49\textwidth} p{0.49\textwidth}}
\hline
\textbf{Symbol} & \textbf{Meaning} \\
\hline
\endfirsthead

\hline
\textbf{Symbol} & \textbf{Meaning} \\
\hline
\endhead

\hline
\multicolumn{2}{r}{Continued on next page}\\
\hline
\endfoot

\hline
\endlastfoot

\multicolumn{2}{l}{\textbf{Core problem data and sets}} \\
$A \in \Zmbb^{m\times n},\; b \in \Zmbb^{m},\; c \in \R{n}$ & Problem data (constraint matrix, RHS, cost) \\
$\Z{n}_{\ge 0},\; \R{n}_{\ge 0}$ & Nonnegative orthants over $\Z{n}$ and $\R{n}$ \\
$\Pcal := \{ x \in \Z{n}_{\ge 0} : Ax = b \}$ & ILP feasible set \\
$\Pcal_{\text{LP}} := \{ x \in \R{n}_{\ge 0} : Ax = b \}$ & LP relaxation feasible set \\
$\OPT := \min\{ c^{\top} x : x \in \Pcal \}$ & ILP optimal value \\
$\OPT_{\text{LP}} := \min\{ c^{\top} x : x \in \Pcal_{\text{LP}} \}$ & LP optimal value \\

\hline
\multicolumn{2}{l}{\textbf{LP bases, partitions, and reduced costs}} \\
$\Bcal \subset [n]$ & LP basis (with $|\Bcal|=m$ if $\rank(A)=m$) \\
$\Ncal := [n]\setminus \Bcal$ & Nonbasic index set \\
$A = [A_{\Bcal}, A_{\Ncal}]$ & Partition of $A$ by $(\Bcal,\Ncal)$ \\
$c = [c_{\Bcal}, c_{\Ncal}],\; x = [x_{\Bcal}, x_{\Ncal}]$ & Matching partitions of $c$ and $x$ \\
$\bar{c}_{\Ncal} := c_{\Ncal} - A_{\Ncal}^{\top}(A_{\Bcal}^{-1})^{\top} c_{\Bcal}$ & Reduced costs for nonbasic indices \\

\hline
\multicolumn{2}{l}{\textbf{Group relaxation (primal formulations)}} \\
$\Pcal_{\Bcal} := \{ (x_{\Bcal}, x_{\Ncal}) \in \Z{m}\times \Z{n-m}_{\ge 0} : Ax=b \}$ & Group relaxation feasible set \\
$\OPT_{\Bcal} := \min\{ c^{\top}x : x \in \Pcal_{\Bcal} \}$ & Group relaxation optimal value \\
$\OPT_{\Bcal} = \min_{x_{\Ncal} \in \Zmbb^{n-m}_{\ge 0}} \{ c_{\Bcal}^{\top} A_{\Bcal}^{-1} b + \bar{c}_{\Ncal}^{\top} x_{\Ncal} \}$ & Objective in nonbasic variables \\ 
%$\widetilde{\OPT}_{\Bcal} := \OPT_{\Bcal} - c_{\Bcal}^{\top} A_{\Bcal}^{-1} b$ & Shifted objective (constant removed) \\

\hline
\multicolumn{2}{l}{\textbf{Smith normal form and bounded formulation}} \\
$A_{\Bcal} = U R V$ & Smith normal form \\
$R = \textup{diag}(r_1, \dots, r_m)$ & Invariant factors of $A_{\Bcal}$ \\
$U^{-1}A_{\Ncal}x_{\Ncal} = U^{-1}b \pmod{R\Z{m}}$ & Feasibility in SNF coordinates \\
% $s_j$ & Order of column $(U^{-1}A_{\Ncal})_j$ in quotient group \\
% $S := \diag(s_1,\dots,s_{d})$ & Diagonal of orders for bounding \\
%$Z$ & Ambient finite product-of-cycles space \\
$\Ambf := U^{-1}A_{\Ncal},\; \bmbf := U^{-1}b,\; \mathbf{c} := \bar{c}_{\Ncal}$ & Data for group formulation \\
%$\widetilde{\OPT}_{\Bcal} = \min_{x \in Z} \{ \mathbf{c}^{\top} x : \Ambf x = \bmbf \pmod{R\Z{m}} \}$ \dyh{come back to, might need to be updated.} & Bounded group formulation \\

\hline
\multicolumn{2}{l}{\textbf{Groups, lattices, kernel and cosets}} \\
$\lat := A \Z{},\; \lat_{\Bcal} := A_{\Bcal} \Z{}$ & Lattices generated by $A$ and $A_{\Bcal}$ \\
$G := (\Ambf\Z{d} +  R\Z{m})/R\Z{m}$ & Range-space quotient (finite abelian group) \\
$\mathcal{R} := \mathbb{Z}^{m}/R\mathbb{Z}^m$ & Ambient range space \\
$Z := \mathbb{Z}_{r_m}^d$ & Ambient space for finite group relaxation \\
$K := \{ \mathbf{h} \in Z : \Ambf \mathbf{h} = \mathbf{0} \pmod{R\Z{m}} \}$ & Null space subgroup  \\
$\hat{\xmbf}$ & Vector in $Z$ with $\Ambf \hat{\xmbf} = \bmbf \pmod{R\Z{m}}$ \\
$\overline{\Pcal}_{\Bcal} = \hat{\xmbf} + K$ & Feasible set (coset)  \\
$\Scal = \{h_1,\dots,h_k\},\; u_j := \lvert \langle h_j\rangle \rvert$ & Generators of $K$ and their orders \\
$K \cong \bigoplus_{j=1}^{k} \Zmbb_{u_j}$ & Cyclic decomposition of $K$ \\
$\Pcal_{\Bcal}^{\star},\; K^{\star}$ & Group-relaxation optima; optimal subset in $K$ \\
\\
\\
\hline
\multicolumn{2}{l}{\textbf{Markov chains and generalized short path (SP) constructs}} \\
$P$ & Transition matrix \\
$\pi$ & Stationary distribution of $P$ \\
$D(P) := \operatorname{diag}(\pi)^{1/2} P \operatorname{diag}(\pi)^{-1/2}$ & Discriminant matrix \\
$\delta$ & Spectral gap of $P$ \\
$t_{\text{mix}}(\varepsilon),\; \textup{TV}(\mu,\nu)$ & Mixing time; total variation distance \\
$H,\; E^{\star}$ & Cost Hamiltonian; ground state energy \\
$H_{\mu} := -D(P) + \mu\, \theta_{\eta}(H/|E^{\star}|)$ & Short path Hamiltonian \\
$\theta_{\eta}(x) := \min\{0,(x+1-\eta)/\eta\}$ & Thresholding function (SP) \\
$\ket{\psi_{\mu}},\; \Pi^{\star}$ & Ground state of $H_{\mu}$; projector onto optima \\
$\|f\|_{P} := \max_{x}\E_{y\sim P(x,\cdot)}[(f(x)-f(y))^2]$ & $P$-pseudo Lipschitz norm \\
$\Delta_{P},\; \gamma,\; \omega$ & Stability; spectral density; log-Sobolev constant \\
\end{longtable}

\end{document}